\documentclass[12pt]{article}
\usepackage[T1]{fontenc}
\usepackage[latin9]{inputenc}
\usepackage{geometry}
\usepackage{amsmath}
\usepackage{amsthm}
\usepackage{amssymb}
\usepackage{graphicx}
\usepackage{setspace}
\usepackage[authoryear]{natbib}
\makeatletter
\theoremstyle{plain}
\newtheorem{assumption}{\protect\assumptionname}
\theoremstyle{plain}
\newtheorem{thm}{\protect\theoremname}
\theoremstyle{remark}
\newtheorem{rem}{\protect\remarkname}

\makeatother

\providecommand{\assumptionname}{Assumption}
\providecommand{\remarkname}{Remark}
\providecommand{\theoremname}{Theorem}

\begin{document}
\title{Identification of time-varying counterfactual parameters in nonlinear
panel models}
\date{\today}
\author{Irene Botosaru and Chris Muris\thanks{McMaster University, Department of Economics. Emails: botosari@mcmaster.ca
and muerisc@mcmaster.ca. We are grateful to an anonymous referee,
and to Xavier D\textquoteright Haultfoeuille, Jiaying Gu, Bo Honore,
Shakeeb Khan, Krishna Pendakur, and Alexandre Poirier for helpful
comments. We thank participants at the Oxford Panel Data Workshop,
the 2023 International Panel Data Conference, CIREQ Econometrics Conference
2023, and at seminars at Georgetown University and the University
of Georgia Athens for comments and suggestions. We gratefully acknowledge
financial support from the Social Sciences and Humanities Research
Council of Canada under grant IG 435-2021- 0778. This research was
undertaken, in part, thanks to funding from the Canada Research Chairs
Program.}}
\maketitle
\begin{abstract}
We develop a general framework for the identification of counterfactual
parameters in a class of nonlinear semiparametric panel models with
fixed effects and time effects. Our method applies to models for discrete
outcomes (e.g., two-way fixed effects binary choice) or continuous
outcomes (e.g., censored regression), with discrete or continuous
regressors. Our results do not require parametric assumptions on the
error terms or time-homogeneity on the outcome equation. Our main
results focus on static models, with a set of results applying to
models without any exogeneity conditions. We show that the survival
distribution of counterfactual outcomes is identified (point or partial)
in this class of models. This parameter is a building block for most
partial and marginal effects of interest in applied practice that
are based on the average structural function as defined by \citet{BlundellPowell2003,BlundellPowell2004}.
To the best of our knowledge, ours are the first results on average
partial and marginal effects for binary choice and ordered choice
models with two-way fixed effects and non-logistic errors.
\end{abstract}

\section{Introduction}

We study counterfactual or policy parameters for nonlinear panel
models with structural equation
\begin{align}
Y_{it}\left(x\right) & =h_{t}\left(\alpha_{i}+x\beta-U_{it}\right),\text{ \ensuremath{t=1,2,\dots,T},}\label{eq:potential_outcome_equation}
\end{align}
where $i$ indexes individual units, $t$ indexes time, $h_{t}$ is
a weakly-monotone transformation function that can vary over $t$
in an unrestricted way, $\beta\in\mathbb{R}^{k}$ is a vector of regression
coefficients, $\alpha_{i}$ is an individual-specific effect, and
$U_{it}$ is a stochastic error term.\footnote{There is a large literature on the identification and estimation of
structural parameters in the cross-sectional version of this model,
e.g., see the work on single index models by \citet{Han1987}, \citet{PowellStockStoker},
\citet{Ichimura}, \citet{Ahn2018}, and references therein.} For each unit $i$, we observe covariates $X_{it}$, $t=1,\dots,T$.
The dependence between $\alpha_{i}$ and $X_{i}=\left(X_{i1},\dots,X_{iT}\right)$
is left unrestricted, so that $\alpha_{i}$ is a fixed effect, c.f.,
e.g., \citet{GrahamPowell2012}. The class of models with outcome
equation as in (\ref{eq:potential_outcome_equation}) includes the
binary choice model with two-way fixed effects, the ordered choice
model with fixed effects and time-varying cut-offs, the censored regression
model with time-varying censoring, and various transformation models
for continuous dependent variables.\footnote{For example, letting $v\equiv\alpha_{i}+x\beta-U_{it}$, the structural
equation for the binary choice model with two-way fixed effects is
$h_{t}\left(v\right)=1\left\{ v\geq\lambda_{t}\right\} $, for the
ordered choice model with time-varying cutoffs it is $h_{t}\left(v\right)=\sum_{j=1}^{J}1\left\{ v\geq\lambda_{jt}\right\} $,
$J\in\mathbb{N}$, and for censored regression with time-varying censoring
it is $h_{t}\left(v\right)=\max\left\{ \lambda_{t},v\right\} $.}

For a subpopulation of individuals defined by their sequence of regressor
values $X_{i}$, our parameter of interest is the counterfactual survival
probability:\footnote{The counterfactual survival probability answers the question ``For
a subpopulation defined by their sequence of regressor values $X_{i}$,
what is the \textit{ceteris paribus} probability that their period-$t$
outcome $Y_{it}$ exceeds $y$ if their period-$t$ regressor values
were exogenously set to $x$?''}
\begin{equation}
\tau_{t,x,y}\left(X_{i}\right)\equiv P\left(\left.Y_{it}\left(x\right)\geq y\right|X_{i}\right),\text{ \ensuremath{t=1,2,\dots,T},}\label{eq:parameter_of_interest}
\end{equation}
where $Y_{it}\left(x\right)$ is given by (\ref{eq:potential_outcome_equation}),
$x$ is a fixed counterfactual value of the period-$t$ regressors,
and $y\in\underline{\mathcal{Y}}\equiv\mathcal{Y}\setminus\text{inf}\mathcal{Y}$
is a fixed cut-off value, where $\mathcal{Y}\subseteq\mathbb{R}$
denotes the support of the observed $Y_{it}=Y_{it}\left(X_{it}\right)$.\footnote{Conditioning on the sequence $X_{i}$ allows us to identify the same
parameter across different exogeneity and time-stationarity assumptions.} The parameter in (\ref{eq:parameter_of_interest}) is a building
block for most partial and marginal effects of interest in applied
practice that are based on the average structural function (ASF) as
defined in the pioneering work of \citet{BlundellPowell2003,BlundellPowell2004}.
For example, when $Y_{it}\left(x\right)$ is non-negative , the ASF
at time $t$ can be obtained as:\footnote{See, for example, \citet{SongWang} for the integrated tail probability
expectation formula that uses a survival function as defined in (\ref{eq:parameter_of_interest}).}
\begin{equation}
ASF_{t}\left(x\right)=\mathbb{E}\left[\mathbb{E}\left(\left.Y_{it}\left(x\right)\right|X_{i}\right)\right]=\mathbb{E}\left(\int_{0}^{\infty}\tau_{t,x,y}\left(X_{i}\right)dy\right).\label{eq:ASF}
\end{equation}
The ASF can then be used to define partial effects based on partial
derivatives or marginal effects based on discrete differences, see,
e.g., \citet{LinWooldridge2015}.

The challenge is to identify (\ref{eq:parameter_of_interest}) in
nonlinear panel models with structural equation as in (\ref{eq:potential_outcome_equation})
when $\alpha_{i}$ are fixed effects and $T<\infty$. To see that
this is challenging, consider the special case of the binary choice
model with two-way fixed effects. A recent literature has made progress
in identifying certain counterfactual parameters for this model provided
that the error terms follow a standard logistic distribution, see
e.g. \citet{aguirregabiriaidentification}, \citet{davezies2021identification},
\citet{dobronyi2021identification}.\footnote{Earlier work by \citet{honore2006boundson} provides partial identification
of marginal effects under a more general structure with dynamics and
arbitrary but known error term distributions. See also \citet{PakelWeidner}
for an approach that applies to parametric models covered by the results
in \citet{bonhomme2012functional}. Finally, see \citet{honore_censored_marginal}
for results on marginal effects for the censored regression model.}

A separate literature provides identification results under time-homogeneity
assumptions that do not allow for arbitrary time-effects, see the
benchmark results in \citet{hoderlein2012nonparametric}, \citet{chernozhukov2013average},
and \citet{chernozhukov2015nonparametric}.\footnote{There, the authors consider a nonseparable structural function and
impose no parametric assumptions on the error terms. When $h_{t}=h$,
the class we models we study here is nested in their analysis.} To the best of our knowledge, nothing is known about the identification
of the ASF for the binary choice model with two way fixed effects
without logistic errors.\footnote{We do not consider here the case of correlated random effects or the
case of large-$T$. Progress on counterfactual parameters for the
former case has been made by, e.g., \citet{ArellanoCarrasco2003},
\citet{altonji2005crosssection}, \citet{bester2009identification},
\citet{ChenKhanTang2019}, \citet{liu2021identification}, while for
the latter by, e.g., \citet{FernandezVal2009}, \citet{FernandezValWeidner2018},
and \citet{bartolucci_partial_effect2023}.}

We derive (partial) identification results for (\ref{eq:parameter_of_interest})
without parametric restrictions on the distribution of $U_{it}$ for
panel models with outcome equation as in (\ref{eq:potential_outcome_equation}).
Our results are for short-$T$. Relevant examples of models to which
our results apply are (i) binary choice with two-way fixed effects
and nonlogistic errors, (ii) ordered choice with time-varying thresholds,
(iii) censored regression with time-varying censoring. Additionally,
since (\ref{eq:parameter_of_interest}) varies over time whenever
$h_{t}$ is time-varying, policy parameters that are functionals of
(\ref{eq:parameter_of_interest}) are also time-varying.

For nonseparable panel models with fixed effects, the results in benchmark
work such as \citet{hoderlein2012nonparametric}, \citet{chernozhukov2013average},
and \citet{chernozhukov2015nonparametric} establish limitations on
what can be learned from panel data in terms of counterfactual parameters.
By imposing additional structure on the latent outcome, such as additivity
in a linear index and the fixed effects, we show that (partial) identification
of counterfactual parameters can be obtained without time-homogeneity
assumptions on the outcome equation and no parametric assumptions
on the distribution of the error terms. Because we make no time-homogeneity
assumptions on $h_{t}$, the counterfactual parameters can vary over
time in an arbitrary way. To the best of our knowledge, it is the
combination of time-varyingness of the outcome equation (hence, of
the counterfactual parameters) and no parametric distributional assumptions
on the error terms that constitutes our contribution relative to the
literature on partial effects in nonlinear panel models with fixed
effects.

For our results, we treat $\beta$ and $h_{t}$ as given, i.e. either
known or previously point- or partially-identified.\footnote{\label{fn:Sufficient-conditions-ID}Sufficient conditions for the
identification of $\beta$ and time-varying $h_{t}$ for models with
structural equation as in (\ref{eq:potential_outcome_equation}) are
provided in \citet{BotosaruMuris2017} and \citet{BotosaruMurisPendakur2021}
under strict exogeneity and weak monotonicity of $h_{t}$, and \citet{BotosaruMurisSokullu2022}
under endogeneity and strict invertibility of $h_{t}$. With a parametric
structure on both $h_{t}$ and the distribution of the stochastic
errors, one may use the results in \citet{bonhomme2012functional}.
Consistent estimators for $\beta$ or/and time-invariant transformation
$h_{t}=h$ in nonlinear panel models without parametric assumptions
on the error terms have been derived by, e.g., \citet{abrevaya1999leapfrog},
\citet{Chen2010}, \citet{ChenWang2018}, \citet{WangChen2020}, \citet{ChenLuWang2022}
and references therein. For specific panel models, such as binary
choice, ordered choice, linear models, duration models, censored regression,
see, e.g., \citet{manski1987semiparametric}, \citet{Honore1992},
\citet{Honore1993}, \citet{HorowitzLee2004}, \citet{ChenDahlKahn},
\citet{Lee2008}, \citet{muris2017estimation}. Recent work on sharp
identification regions for structural parameters includes \citet{KhanPonomarevaTamer2011},
\citet{KhanPonomarevaTamer2016}, \citet{HonoreHu2020}, \citet{KhanPonomarevaTamer2021},
and \citet{aristodemou2021semiparametric}. See also \citet{Ghanem2017}
on testing identifying assumptions in the class of models we consider
here.} Our key insight is that the linear index structure allows us to classify
each observed probability at time $s$,
\begin{equation}
P\left(\left.Y_{is}\geq y^{\prime}\right|X_{i}\right),\;s=1,\dots,T,y^{\prime}\in\underline{\mathcal{Y}},\label{eq:obsats}
\end{equation}
as either an upper bound on $\tau_{t,x,y}\left(X_{i}\right)$, a lower
bound on $\tau_{t,x,y}\left(X_{i}\right)$, or both. We show that
(\ref{eq:obsats}) is an upper bound only if the \emph{observed index
at time $s$,} $X_{is}\beta-h_{s}^{-}\left(y\right)$, is at least
the \emph{counterfactual index }at time $t$, $x\beta-h_{t}^{-}\left(y\right)$;
and a lower bound only if the observed index at time $s$ is at most
the counterfactual index at time $t$. Without any additional time-stationarity
or exogeneity assumptions, bounds on the period-$t$ counterfactual
probability $\tau_{t,x,y}\left(X_{i}\right)$ can be constructed from
the period-$t$ observed probabilities $P\left(\left.Y_{it}\geq y^{\prime}\right|X_{i}\right)$,
$y^{\prime}\in\underline{\mathcal{Y}}$. Under a conditional time-stationarity
assumption on the errors, outcomes from \emph{all periods }are informative
for the period-$t$ counterfactual probability. The bounds under conditional
time-stationarity are tighter than those without exogeneity assumptions.
Point identification is obtained when the transformation function
$h_{t}$ is invertible or when the counterfactual index equals one
of the observed indices.

The remainder of this paper is organized as follows. Section \ref{sec:model_and_results}
presents our main results: Section \ref{subsec:No-exogeneity-assumptions}
constructs bounds without any additional time-stationarity or exogeneity
assumptions, while Section \ref{subsec:Conditional-time-stationarity}
constructs bounds under a conditional time-stationarity assumption
on the error terms. Section \ref{sec:Examples} applies our results
to a few examples: binary choice model, ordered choice model, and
censored regression. We present a numerical experiment for the binary
choice model in Section \ref{subsec:Two-way-binary-choice-contX}
and one for the ordered choice model in Section \ref{subsec:Staggered-adoption-with}.
All proofs and an additional numerical experiment can be found in
the Appendix.

\section{Main results\label{sec:model_and_results}}

We provide two results on the identification of $\tau_{t,x,y}$ defined
in (\ref{eq:parameter_of_interest}). Our first result in Theorem
\ref{thm:insight_no_exogeneity} provides bounds on $\tau_{t,x,y}$
without imposing any exogeneity assumptions or time-stationarity assumptions
on $U_{it},\,X_{i}$. Our second result in Theorem \ref{thm:intersection_bound_strict_exogeneity}
uses a conditional time-stationarity assumption on $\left.U_{it}\right|\alpha_{i},X_{i}$
to tighten those bounds. Conditioning on the sequence $X_{i}$ allows
us to identify the same parameter, $\tau_{t,x,y}$, across different
exogeneity and time-stationarity assumptions.

The following two assumptions are maintained throughout the paper:
\begin{assumption}
\label{assu:observability}(i) The distribution of panel data $\left(X_{i1},\cdots,X_{iT},Y_{i1},\cdots,Y_{iT}\right)$
is observed, where $Y_{it}=Y_{it}\left(X_{it}\right)$ is generated
by (\ref{eq:potential_outcome_equation}); (ii) $\beta$ and $h_{t}$
in (\ref{eq:potential_outcome_equation}) are either known, point-identified,
or partially-identified.
\end{assumption}
Footnote \ref{fn:Sufficient-conditions-ID} lists work that provides
sufficient assumptions for either the point- or the partial-identification
of $\beta$ and $h_{t}$.
\begin{assumption}
\label{assu:weak_mono} For each $t$, $h_{t}:\mathbb{R}\rightarrow\mathcal{Y}\subseteq\mathbb{R}$
is weakly-monotone and right-continuous.
\end{assumption}
Assumption \ref{assu:weak_mono} does not restrict the way that $h_{t}$
can vary over $t$. Our results apply to the case of time-invariant
$h_{t}=h$, in which case parameters such as (\ref{eq:parameter_of_interest})
and (\ref{eq:ASF}) are also time-invariant. This assumption allows
$h_{t}$ to have flat parts and jumps, or to be continuous. Hence,
our setting accommodates both discrete and continuous outcomes.

Given Assumption \ref{assu:weak_mono}, we define the generalized
inverse of $h_{t}$ as:\footnote{Existence of $h_{t}^{-}$ is ensured by Assumption \ref{assu:weak_mono}.}
\begin{equation}
h_{t}^{-}\left(y\right)\equiv\inf\left\{ y^{*}:h_{t}\left(y^{*}\right)\geq y\right\} ,y\in\mathcal{\underline{\mathcal{Y}}},\label{eq:generalizedinverse}
\end{equation}
i.e. it is the smallest value of the latent variable that yields a
value of the observed outcome $Y_{t}\geq y$.

For what follows we fix $t,x$ and we fix $y\in\underline{\mathcal{Y}}$.

\subsection{No additional exogeneity or time-stationarity assumptions\label{subsec:No-exogeneity-assumptions}}

For our first set of results, we define the following sets:

\begin{align}
\mathcal{Y}_{L} & \equiv\left\{ y^{\prime}\in\underline{\mathcal{Y}}:X_{it}\beta-h_{t}^{-}\left(y^{\prime}\right)\leq x\beta-h_{t}^{-}\left(y\right)\right\} ,\label{eq:Y_low}\\
\mathcal{Y}_{U} & \equiv\left\{ y^{\prime}\in\underline{\mathcal{Y}}:X_{it}\beta-h_{t}^{-}\left(y^{\prime}\right)\geq x\beta-h_{t}^{-}\left(y\right)\right\} .\label{eq:Y_up}
\end{align}
The set $\mathcal{Y}_{L}$ collects the values $y^{\prime}\in\underline{\mathcal{Y}}$
for which the \emph{counterfactual index at time $t$,} $x\beta-h_{t}^{-}\left(y\right)$,
is greater than the \emph{observed index }at time $t$, $X_{it}\beta-h_{t}^{-}\left(y^{\prime}\right)$,
while $\mathcal{Y}_{U}$ collects the values $y^{\prime}\in\underline{\mathcal{Y}}$
for which the counterfactual index at time $t$ is smaller than the
observed index at time $t$. Note that $\mathcal{Y}_{L}\cup\mathcal{Y}_{U}=\underline{\mathcal{Y}}$.
\begin{thm}
\label{thm:insight_no_exogeneity}Let $Y_{it}$ follow (\ref{eq:potential_outcome_equation})
and let Assumptions \ref{assu:observability} and \ref{assu:weak_mono}
hold. Then, for given values of $\beta$ and $h_{t}$,
\begin{equation}
\tau_{t,x,y}\left(X_{i}\right)\in\left[\sup_{y^{\prime}\in\mathcal{Y}_{L}}P\left(\left.Y_{it}\geq y^{\prime}\right|X_{i}\right),\inf_{y^{\prime}\in\mathcal{Y}_{U}}P\left(\left.Y_{it}\geq y^{\prime}\right|X_{i}\right)\right]\cap\left[0,1\right],\label{eq:bounds_statement}
\end{equation}
using the convention that $\sup\emptyset=-\infty$ and $\inf\emptyset=+\infty$.
\end{thm}
Theorem \ref{thm:insight_no_exogeneity} uses the linear-index structure
of (\ref{eq:potential_outcome_equation}) and knowledge of $\beta$
and $h_{t}$ to classify $P\left(\left.Y_{it}\geq y^{\prime}\right|X_{i}\right)$
as a lower (upper) bound on $\tau_{t,x,y}\left(X_{i}\right)$ for
any value of $y^{\prime}\in\underline{\mathcal{Y}}$. By varying $y^{\prime}\in\underline{\mathcal{Y}},$
we obtain the sets $\mathcal{Y}_{L}$ ($\mathcal{Y}_{U})$ of values
$y^{\prime}$ that provide lower (upper) bounds. Equation (\ref{eq:bounds_statement})
intersects these bounds. Since $\mathcal{Y}_{L}\cup\mathcal{Y}_{U}=\underline{\mathcal{Y}}$,
every value of $y^{\prime}$ provides an upper bound, a lower bound,
or both. Note that since either $\mathcal{Y}_{L}$ or $\mathcal{Y}_{U}$
can be empty, the trivial lower (upper) bound is selected by the intersection
with $\left[0,1\right]$.

Section \ref{sec:Examples} applies Theorem \ref{thm:insight_no_exogeneity}
to binary choice, ordered choice, and censored regression. In the
case of binary choice, the bounds of Theorem \ref{thm:insight_no_exogeneity}
simplify significantly. In that case, $P\left(\left.Y_{it}\geq1\right|X_{i}\right)$
provides an upper bound for $P\left(\left.Y_{it}\left(x\right)\geq1\right|X_{i}\right)$
if $X_{it}\beta>x\beta$, and a lower bound if $X_{it}\beta<x\beta$.
When $X_{it}\beta=x\beta$, $P\left(\left.Y_{it}\geq1\right|X_{i}\right)$
point-identifies $P\left(\left.Y_{it}\left(x\right)\geq1\right|X_{i}\right)$.
\begin{rem}
The result in Theorem \ref{thm:insight_no_exogeneity} is stated for
a given value of $\left(\beta,h_{t}\right)$. Thus, the bounds are
directly applicable in case where $\left(\beta,h_{t}\right)$ are
point-identified. If $\left(\beta,h_{t}\right)$ are partially identified,
bounds on $\tau_{t,x,y}\left(X_{i}\right)$ can be obtained by taking
the worst-case bounds across parameter values in the identified set.
\end{rem}
\medskip{}

\begin{rem}
The bounds in Theorem \ref{thm:insight_no_exogeneity} do not require
any additional exogeneity conditions on $\left(X_{i},U_{i}\right)$
beyond those that may be needed for Assumption (\ref{assu:observability})(ii).
In particular, the bounds are valid for models with strictly or weakly
exogenous regressors, with lagged dependent variables, and with endogenous
regressors. Likewise, the bounds do not require any time-stationarity
assumptions on the distribution of the error terms. In this sense,
the bounds are valid for models with errors with time-varying distributions.
Finally, separability in $\alpha_{i}$ and $U_{it}$ is not necessary
for identification of the counterfactual survival probability, suggesting
that our argument can be applied to a more general class of models.
\end{rem}
\medskip{}

\begin{rem}
\label{rem:Point-identification-happens}Point identification occurs
in a number of settings. First, when $h_{t}$ is invertible, see Section
(\ref{subsec:invertible_h}). Second, when there exists a $y^{\prime}$
such that
\[
h_{t}^{-}\left(y^{\prime}\right)=\left(X_{it}-x\right)\beta+h_{t}^{-}\left(y\right).
\]
This can happen, among others, when $y=y^{\prime}$ and $X_{it}=x$
(or $X_{it}\beta=x\beta$). In this case, $y\in\mathcal{Y}_{L}\cap\mathcal{Y}_{U}$
and the upper and lower bounds coincide.
\end{rem}
\medskip{}

\begin{rem}
\label{rem:cardinality}The bounds in Theorem \ref{thm:insight_no_exogeneity}
shrink with the cardinality of $\underline{\mathcal{Y}}$, so that
the worst case for our bounds is when the outcomes are binary (worst
case here means that, provided failure of point-identification, one
of the bounds is always trivial), while the best case is when the
outcomes are continuous (best case in the sense that the parameter
of interest is always point-identified).
\end{rem}
\medskip{}

\subsection{Conditional time-stationarity\label{subsec:Conditional-time-stationarity}}

The bounds in Theorem \ref{thm:insight_no_exogeneity} can be wide,
for example when the dependent variable has few points of support
(see Remark \ref{rem:cardinality}). The following assumption obtains
tighter bounds by using information across all time periods rather
than information from period $t$ only.
\begin{assumption}
\label{assu:Strict-exogeneity}Conditional time-stationarity: $\left.U_{it}\right|\alpha_{i},X_{i}\stackrel{d}{=}\left.U_{i1}\right|\alpha_{i},X_{i}$,
for all $t=2,\dots,T$.
\end{assumption}
Conditional time-stationarity requires that the conditional distribution
of the error terms conditional on $\alpha_{i},X_{i}$ be the same
in each time period. In \citet{chernozhukov2013average,chernozhukov2015nonparametric}
and \citet{hoderlein2012nonparametric}\footnote{See, e.g., \citet{manski1987semiparametric}, \citet{Honore1992},
\citet{Abrevaya2000}, \citet{hoderlein2012nonparametric}, \citet{GrahamPowell2012},
\citet{chernozhukov2013average}, \citet{chernozhukov2015nonparametric},
\citet{KhanPonomarevaTamer2016}, \citet{ChenKhanTang2019}, \citet{KhanPonomarevaTamer2021}.} the authors refer to this assumption as both ``strict exogeneity''
and time-homogeneity of the error terms.\footnote{For a discussion of strict exogeneity, as well as other notions of
exogeneity, in the context of linear models, see \citet{Chamberlain1984},
\citet{ArellanoHonore2001}, \citet{ArellanoBonhomme2011}.} Note that the error terms are required to have a time-stationary
(``time-homogeneous'') distribution conditional on $\alpha_{i}$
and the entire sequence $X_{i1},\dots,X_{iT}$. The assumption allows
for serial correlation in the errors $U_{it}$ and in some components
of $X_{i}$, and it leaves the distribution of $\alpha_{i}$ conditional
on $X_{i}$ unrestricted. Note that, in our set-up with time-varying
structural equation, the conditional distribution of $Y_{it}$ given
$X_{i}$ can still vary over time.

Fix $t,x,$$y\in\underline{\mathcal{Y}}$ and define the sets:
\begin{align*}
\mathcal{L} & \equiv\left\{ \left(s,y^{\prime}\right)\in\left\{ 1,\cdots,T\right\} \times\mathcal{\underline{\mathcal{Y}}}:X_{is}\beta-h_{s}^{-}\left(y^{\prime}\right)\leq x\beta-h_{t}^{-}\left(y\right)\right\} ,\\
\mathcal{U} & \equiv\left\{ \left(s,y^{\prime}\right)\in\left\{ 1,\cdots,T\right\} \times\underline{\mathcal{Y}}:X_{is}\beta-h_{s}^{-}\left(y^{\prime}\right)\geq x\beta-h_{t}^{-}\left(y\right)\right\} .
\end{align*}

\begin{thm}
\label{thm:intersection_bound_strict_exogeneity}Let $Y_{i1},\cdots,Y_{iT}$
follow (\ref{eq:potential_outcome_equation}), and let Assumptions
(\ref{assu:observability}), \ref{assu:weak_mono}, and \ref{assu:Strict-exogeneity}
hold. Then, for given values of $\beta$ and $h_{t}$,
\[
\tau_{t,x,y}\left(X_{i}\right)\in\left[\sup_{\left(s,y^{\prime}\right)\in\mathcal{L}}P\left(\left.Y_{is}\geq y^{\prime}\right|X_{i}\right),\inf_{\left(s,y^{\prime}\right)\in\mathcal{U}}P\left(\left.Y_{is}\geq y^{\prime}\right|X_{i}\right)\right]\cap\left[0,1\right],
\]
using the convention that $\sup\emptyset=-\infty$ and $\inf\emptyset=+\infty$.
\end{thm}
As in the case of in Theorem \ref{thm:insight_no_exogeneity}, the
result above applies directly to point-identified $\beta$ and $h_{t}$.
If $\left(\beta,h_{t}\right)$ are partially identified, bounds on
$\tau_{t,x,y}\left(X_{i}\right)$ can be obtained by taking the worst-case
bounds across parameter values in the identified set. 

According to Theorem \ref{thm:intersection_bound_strict_exogeneity},
we can use information from any period $s$ to construct bounds on
the counterfactual survival distribution in period $t$. The resulting
bounds can be much more informative than those in Theorem \ref{thm:insight_no_exogeneity}
without time-stationarity and exogeneity assumptions; this can be
seen from the expressions for binary and ordered choice in Section
\ref{sec:Examples}, and from the numerical experiments in Section
\ref{subsec:Numerical-example}. The gains can be substantial, especially
if the number of time periods is large, if there is variation in the
values of the sequence $X_{i}$, and if there is a large degree of
variation over time in $h_{t}$.\footnote{Although, we do not prove sharpness of the bounds in Theorem \ref{thm:intersection_bound_strict_exogeneity},
we were unable to construct a situation where they were not. As we
show in the numerical exercises, the bounds are (very) informative.}
\begin{rem}
That variation in $h_{t}$ can improve identification is particularly
interesting. In related settings, time-homogeneity of the outcome
equation, i.e. $h_{t}=h$, has been used for identification of partial
effects in panel models. In our setting, time-variation in the structural
function can aid identification of partial effects in monotone single-index
models. For an example, see our analysis of binary choice models in
Section \ref{sec:Panel-binary-choice}.
\end{rem}
\medskip{}

\begin{rem}
\label{rem:alternative_strict_exo_bound_formulation}The best upper
and lower bounds can be thought of as, first, choosing the best $y^{\prime}$
for each time period (as in Theorem \ref{thm:insight_no_exogeneity})
and then choosing the best time period. Letting 
\begin{align*}
\mathcal{Y}_{sL} & \equiv\left\{ y^{\prime}\in\mathcal{\underline{\mathcal{Y}}}:X_{is}\beta-h_{s}^{-}\left(y^{\prime}\right)\leq x\beta-h_{t}^{-}\left(y\right)\right\} ,\\
\mathcal{Y}_{sU} & \equiv\left\{ y^{\prime}\in\underline{\mathcal{Y}}:X_{is}\beta-h_{s}^{-}\left(y^{\prime}\right)\geq x\beta-h_{t}^{-}\left(y\right)\right\} ,
\end{align*}
we obtain a bound for each time period,
\[
\tau_{t,x,y}\left(X_{i}\right)\in\left[\sup_{y^{\prime}\in\mathcal{Y}_{sL}}P\left(\left.Y_{is}\geq y^{\prime}\right|X_{i}\right),\inf_{y^{\prime}\in\mathcal{Y}_{sU}}P\left(\left.Y_{is}\geq y^{\prime}\right|X_{i}\right)\right]\cap\left[0,1\right]
\]
and, because this bound applies for each period $s$, we obtain 
\[
\tau_{t,x,y}\left(X_{i}\right)\in\cap_{s\in\left\{ 1,\cdots,T\right\} }\left[\sup_{y^{\prime}\in\mathcal{Y}_{sL}}P\left(\left.Y_{is}\geq y^{\prime}\right|X_{i}\right),\inf_{y^{\prime}\in\mathcal{Y}_{sU}}P\left(\left.Y_{is}\geq y^{\prime}\right|X_{i}\right)\right]\cap\left[0,1\right]
\]
or 
\[
\max_{s\in\left\{ 1,\cdots,T\right\} }\sup_{y^{\prime}\in\mathcal{Y}_{sL}}P\left(\left.Y_{is}\geq y^{\prime}\right|X_{i}\right)\leq\tau_{t,x,y}\left(X_{i}\right)\leq\min_{s\in\left\{ 1,\cdots,T\right\} }\inf_{y^{\prime}\in\mathcal{Y}_{sU}}P\left(\left.Y_{is}\geq y^{\prime}\right|X_{i}\right).
\]

We illustrate this for the binary and ordered choice models in Section
\ref{sec:Examples}.
\end{rem}
\medskip{}

\begin{rem}
As in the case of in Theorem \ref{thm:insight_no_exogeneity}, the
bounds in Theorem \ref{thm:intersection_bound_strict_exogeneity}
do not require separability in $\alpha_{i}$ and $U_{it}$ for identification
of the counterfactual survival probability. Thus, our argument may
be applied to a more general class of models.
\end{rem}
\medskip{}

\begin{rem}
Theorem \ref{thm:intersection_bound_strict_exogeneity} updates and
replaces the results in Section 3.2 in Botosaru and Muris (2017).
\end{rem}

\section{Examples\label{sec:Examples}}

In this section, we apply our results to a number of empirically relevant
choices for $h_{t}$. These examples are helpful in understanding
how informative our bounds are, and help relate them to existing bounds
in the literature derived under related but different conditions.

Let $v\equiv\alpha_{i}+X_{it}\beta-U_{it}$. We start by applying
our results to the fixed effects linear transformation model with
invertible $h_{t}$. We then study binary choice models with two-way
fixed effects $h_{t}\left(v\right)=1\left\{ v\geq\lambda_{t}\right\} $,
ordered choice models with time-varying thresholds $h_{t}\left(v\right)=\sum_{j=1}^{J}1\left\{ v\geq\lambda_{jt}\right\} $,
$J\in\mathbb{N},$ and censored regression with $h_{t}\left(v\right)=\max\left\{ \lambda_{t},v\right\} $.

\subsection{Continuous outcomes and invertible $h_{t}$\label{subsec:invertible_h}}

Let $h_{t}$ be invertible, so that $h_{t}^{-}=h_{t}^{-1}$. Examples
include linear regression with two-way fixed effects, i.e. $h_{t}\left(v\right)=v+\lambda_{t}$;
transformation models used in duration analysis; and the Box-Cox transformation
model.\footnote{Such models have been studied extensively in the cross-sectional setting,
see, e.g., \citet{amemiyaComparisonBoxCoxMaximum1981,barnettNonparametricSemiparametricMethods1991,powellRescaledMethodsofmomentsEstimation1996}.}

Let $h_{t}$ be defined on $\mathbb{R},$or $v\in\mathbb{R}.$ There
always exists a $y^{\prime}$ such that
\begin{align*}
y^{\prime} & =h_{t}\left(h_{t}^{-1}\left(y\right)+\left(x-X_{it}\right)\beta\right)\\
 & \in\mathcal{Y}_{L}\cap\mathcal{Y}_{U}.
\end{align*}
Theorem \ref{thm:insight_no_exogeneity} applies and the counterfactual
survival probability (\ref{eq:parameter_of_interest}) is point-identified.
To see this, consider the argument below:\footnote{See also Remark 2 in Botosaru and Muris (2017), and Botosaru et al.
(2021).}
\begin{align*}
\tau_{t,x,y}\left(X_{i}\right) & =P\left(\left.Y_{it}\left(x\right)\geq y\right|X_{i}\right)\\
 & =P\left(\left.\alpha_{i}-U_{it}\geq h_{t}^{-1}\left(y\right)-x\beta\right|X_{i}\right)\\
 & =P\left(\left.\alpha_{i}+X_{it}\beta-U_{it}\geq h_{t}^{-1}\left(y\right)+\left(X_{it}-x\right)\beta\right|X_{i}\right)\\
 & =P\left(\left.Y_{it}\geq h_{t}\left(h_{t}^{-1}\left(y\right)+\left(X_{it}-x\right)\beta\right)\right|X_{i}\right)\\
 & =P\left(\left.Y_{it}\geq y^{\prime}\right|X_{i}\right),
\end{align*}
where invertibility of $h_{t}$ was used in the second equality.

We are not aware of results for $\tau_{t,x,y}$ -- or derived quantities
such as partial/marginal effects -- for the case considered here,
other than those in \citet{BotosaruMuris2017,BotosaruMurisPendakur2021,BotosaruMurisSokullu2022}.

\subsection{Binary choice\label{sec:Panel-binary-choice}}

The link function 
\[
h_{t}\left(v\right)=1\left\{ v-\lambda_{t}\geq0\right\} 
\]
obtains the panel binary choice model with two-way fixed effects with
structural function
\begin{equation}
Y_{it}\left(x\right)=1\left\{ \alpha_{i}+x\beta-U_{it}-\lambda_{t}\geq0\right\} .\label{eq:BC_structuralfct}
\end{equation}

Our results yield bounds on (\ref{eq:parameter_of_interest}), hence
on partial effects, without parametric assumptions on the distribution
of the error terms and for a variety of exogeneity conditions. The
only existing results without parametric assumptions on the error
term that we are aware of are those in \citet{chernozhukov2013average},
but those require that regressors be discrete and that $\lambda_{t}=0$
for all $t$.

In what follows, we fix the time period for the counterfactual to
$t=1$ and use the abbreviated notation 
\begin{equation}
\tau_{x}\left(X_{i}\right)\equiv\tau_{1,x,1}\left(X_{i}\right)=P\left(\left.Y_{i1}\left(x\right)\geq1\right|X_{i}\right).\label{eq:BC_survival}
\end{equation}
For this model, $\underline{\mathcal{Y}}=\left\{ 1\right\} $ and
$h_{t}^{-}\left(1\right)=\lambda_{t}$.
\begin{rem}
\label{rem:no_inf}By removing $\text{inf\ensuremath{\mathcal{Y}=0}}$,
we remove $h_{t}^{-}\left(0\right)=-\infty$. This is useful, since,
e.g., Theorem \ref{thm:insight_no_exogeneity} asks us to compare
$X_{it}\beta-h_{t}^{-}\left(y^{\prime}\right)$ to $x\beta-h_{t}^{-}\left(y\right)$
for all values of $y^{\prime}$ and fixed $x,y,X_{i}$. For the particular
case that $y^{\prime}=\text{inf}\ensuremath{\mathcal{Y}=0}=y$, this
obtains $X_{it}\beta+\infty$ and $x\beta+\infty$. Note that we could
allow $y^{\prime}\in\mathcal{Y}$ for fixed $y\in\mathcal{\underline{Y}}$,
but this would lead to obtaining trivial bounds, i.e. for $y=1$ and
$y^{\prime}=0$, $\text{\ensuremath{\tau_{t,x,1}\left(X_{i}\right)}}\leq P\left(\left.Y_{it}\geq0\right|X_{i}\right)=1$.
Hence, we restrict both $y,y^{\prime}\in\underline{\mathcal{Y}}$.
\end{rem}

\subsubsection{Theorem 1 bounds}

For $y^{\prime}=1=y,$ Theorem \ref{thm:insight_no_exogeneity} leads
to three different cases depending on the sign of $X_{i1}\beta-x\beta$.
These cases are:
\begin{enumerate}
\item $X_{i1}\beta=x\beta$, in which case (\ref{eq:BC_survival}) is point-identified;
\item $X_{i1}\beta>x\beta$, in which case the observed index exceeds the
counterfactual one, so that $y^{\prime}=1\in\mathcal{Y}_{U}$ and
the observed probability provides an upper bound for (\ref{eq:BC_survival});
\item $X_{i1}\beta<x\beta$, in which case the observed index is lower than
the counterfactual one, so that $y^{\prime}=1\in\mathcal{Y}_{L}$
and the observed probability provides a lower bound for (\ref{eq:BC_survival}).
\end{enumerate}
The results of Theorem \ref{thm:insight_no_exogeneity} can be summarized
as follows:
\[
\tau_{x}\left(X_{i}\right)\in\begin{cases}
\left[0,\text{min}\left\{ 1,P\left(\left.Y_{i1}\geq1\right|X_{i}\right)\right\} \right], & \text{if }X_{i1}\beta>x\beta,\\
\left\{ P\left(\left.Y_{i1}\geq1\right|X_{i}\right)\right\} , & \text{if }X_{i1}\beta=x\beta,\\
\left[\text{max}\left\{ 0,P\left(\left.Y_{i1}\geq1\right|X_{i}\right)\right\} ,1\right], & \text{if }X_{i1}\beta<x\beta.
\end{cases}
\]

To see that these bounds are valid, we adapt the key derivation underlying
Theorem \ref{thm:insight_no_exogeneity} to this specific case:
\begin{align}
\tau_{x}\left(X_{i}\right) & =P\left(\left.Y_{i1}\left(x\right)\geq1\right|X_{i}\right)\nonumber \\
 & =P\left(\left.\alpha_{i}+x\beta-U_{i1}\geq\lambda_{1}\right|X_{i}\right)\nonumber \\
 & \lesseqqgtr P\left(\left.\alpha_{i}+X_{i1}\beta-U_{i1}\geq\lambda_{1}\right|X_{i}\right)\nonumber \\
 & =P\left(\left.Y_{i1}\geq1\right|X_{i}\right),\label{eq:bc_point_identified}
\end{align}
where the third line denotes that the direction of the inequality
depends on the sign of the difference between the observed index $X_{i1}\beta$
and the counterfactual index $x\beta$.

\subsubsection{Theorem 2 bounds}

Under Assumption \ref{assu:Strict-exogeneity}, Theorem \ref{thm:intersection_bound_strict_exogeneity}
implies that any period can be used to construct counterfactuals for
period $1$. Instead of stating the bounds as implied by Theorem \ref{thm:intersection_bound_strict_exogeneity},
we show how to construct them from first principles.

Suppose that there exists a time period $s$ such that 
\begin{equation}
X_{is}\beta-\lambda_{s}=x\beta-\lambda_{1}.\label{eq:equality}
\end{equation}
Then (\ref{eq:BC_survival}) is point identified with
\begin{align*}
\tau_{x}\left(X_{i}\right) & =P\left(\left.\alpha_{i}+x\beta-\lambda_{1}-U_{i1}\geq0\right|X_{i}\right)\\
 & =P\left(\left.\alpha_{i}+x\beta-\lambda_{1}-U_{is}\geq0\right|X_{i}\right)\\
 & =P\left(\left.\alpha_{i}+X_{is}\beta-\lambda_{s}-U_{is}\geq0\right|X_{i}\right)\\
 & =P\left(\left.Y_{is}\geq1\right|X_{i}\right),
\end{align*}
where the second equality follows by Assumption (\ref{assu:Strict-exogeneity})
and the third equality follows by (\ref{eq:equality}).\footnote{As a special case, in a model without time dummies and with a binary
treatment indicator, e.g., $X_{it}\in\left\{ 0,1\right\} $, we can
point-identify the distribution under treatment $x=1$ for any subpopulation
that is treated at some point, $\exists t:X_{it}=1$.}

If there does not exist a time period such that (\ref{eq:equality})
holds, Theorem \ref{thm:intersection_bound_strict_exogeneity} can
be operationalized as follows. Fix $y^{\prime}=1$ and, for each period
$s\in\left\{ 1,\dots,T\right\} $, compare $X_{is}\beta-\lambda_{s}$
to $x\beta-\lambda_{1}$, and group the time periods according to
whether they provide an upper bound $\left(s\in\mathcal{T}_{U}\right)$
or a lower bound $\left(s\in\mathcal{T}_{L}\right)$:
\begin{align*}
\mathcal{T}_{U} & \equiv\left\{ s\in\left\{ 1,\cdots,T\right\} :X_{is}\beta-\lambda_{s}\geq x\beta-\lambda_{1}\right\} ,\\
\mathcal{T}_{L} & \equiv\left\{ s\in\left\{ 1,\cdots,T\right\} :X_{is}\beta-\lambda_{s}\leq x\beta-\lambda_{1}\right\} .
\end{align*}
These sets correspond to $\mathcal{U}$ and $\mathcal{L}$ in Theorem
\ref{thm:intersection_bound_strict_exogeneity} with $y^{\prime}=1=y$,
counterfactual period $t=1$, and $h_{t}^{-}\left(1\right)=\lambda_{t}.$

The best lower (upper) bound on (\ref{eq:BC_survival}) is constructed
using periods in $\mathcal{T}_{L}$ ($\mathcal{T}_{U}$):
\begin{align}
\max\left\{ 0,\left(P\left(\left.Y_{is}\geq1\right|X_{i}\right)\right)_{s\in\mathcal{T}_{L}}\right\}  & \leq\tau_{x}\left(X_{i}\right)\label{eq:bounds_strict_exo_bc}\\
 & \leq\min\left\{ 1,\left(P\left(\left.Y_{is}\geq1\right|X_{i}\right)\right)_{s\in\mathcal{T}_{U}}\right\} ,
\end{align}
or
\[
\tau_{x}\left(X_{i}\right)\in\begin{cases}
\left[0,\text{min}\left\{ 1,\left\{ P\left(\left.Y_{is}\geq1\right|X_{i}\right)\right\} _{s\in\mathcal{T}_{U}}\right\} \right], & \text{if }X_{is}\beta-\lambda_{s}>x\beta-\lambda_{1},\\
\left\{ P\left(\left.Y_{is}\geq1\right|X_{i}\right)\right\} , & \text{if }X_{is}\beta-\lambda_{s}=x\beta-\lambda_{1},\\
\left[\text{max}\left\{ 0,\left\{ P\left(\left.Y_{is}\geq1\right|X_{i}\right)\right\} _{s\in\mathcal{T}_{L}}\right\} ,1\right], & \text{if }X_{is}\beta-\lambda_{s}<x\beta-\lambda_{1}.
\end{cases}
\]

These bounds have a few interesting properties. First, they can be
informative for stayers, i.e. even when $X_{it}=x$ for all $t$,
nontrivial bounds can be derived as long as there are time effects
$\lambda_{t}\neq\lambda_{1}$ for some $t$. Second, under Assumption
\ref{assu:Strict-exogeneity}, the bounds tighten as compared to the
bounds without this assumption. In Section \ref{subsec:Numerical-example},
we investigate these and other properties through a numerical experiment.

\subsection{Ordered choice\label{subsec:Example:-ordered-choice}}

The link function 
\[
h_{t}\left(v\right)=\sum_{j=1}^{J}1\left\{ v\geq\lambda_{jt}\right\} ,J\in\mathbb{N},
\]
obtains an ordered choice model with fixed effects and time-varying
thresholds:\footnote{\citet{dasPanelDataModel1999,johnsonphdthesis,baetschmannIdentificationEstimationThresholds2012,baetschmann2012identification,muris2017estimation,BotosaruMurisPendakur2021}
discuss identification of the $\left(\beta,\lambda_{tj}\right)$ under
various conditions. With logistic errors, the parameters can be estimated
via composite conditional maximum likelihood estimation. Without logistic
errors, the parameters can be estimated using maximum score methods.
In the logistic case, the results in \citet{davezies2021identification}
can then be used to bound average marginal effects.}
\begin{equation}
h_{t}\left(v\right)=\begin{cases}
1 & \text{ if }-\infty<v<\lambda_{t2},\\
2 & \text{ if }\lambda_{t2}\leq v<\lambda_{t3},\\
\vdots\\
J & \text{ if }\lambda_{tJ}\leq v<+\infty.
\end{cases}\label{eq:ordered_h}
\end{equation}

For this model, $\mathcal{\underline{Y}}=\left\{ 2,\dots,J\right\} $
and $h_{t}^{-}\left(y\right)=\lambda_{ty},\,y\in\underline{\mathcal{Y}}.$

We fix the time period for the counterfactual to $t=1$. The parameter
of interest is then 
\begin{equation}
\tau_{1,x,y}\left(X_{i}\right)=P\left(\left.Y_{i1}\left(x\right)\geq y\right|X_{i}\right).\label{eq:OC_survival}
\end{equation}

\subsubsection{Theorem 1 bounds}

Fix $\left(y,x,X_{i}\right)$. To bound $\tau_{1,x,y}\left(X_{i}\right)$
in (\ref{eq:OC_survival}), for each $y^{\prime}\in\mathcal{\underline{Y}}$
we compare the observed index $X_{i1}\beta-\lambda_{1y^{\prime}}$
to the counterfactual index $x\beta-\lambda_{1y}$, and construct
the sets 
\begin{align*}
\mathcal{Y}_{U} & =\left\{ y^{\prime}\in\mathcal{\underline{Y}}:X_{i1}\beta-\lambda_{1y^{\prime}}\geq x\beta-\lambda_{1y}\right\} ,\\
\mathcal{Y}_{L} & =\left\{ y^{\prime}\in\mathcal{\underline{Y}}:X_{i1}\beta-\lambda_{1y^{\prime}}\leq x\beta-\lambda_{1y}\right\} .
\end{align*}
These sets can be used to form lower/upper bounds on $\tau_{1,x,y}\left(X_{i}\right)$
according to Theorem \ref{thm:insight_no_exogeneity}.

Note that for the ordered choice model, for a fixed $y$, there may
be multiple $y^{'}$s yielding nontrivial bounds. This is different
from the binary choice case where $\underline{\mathcal{Y}}=\left\{ 1\right\} $
has one element, that provides either a non-trivial lower bound \emph{or}
a non-trivial upper bound. For the ordered choice model, $\underline{\mathcal{Y}}$
has at least two elements, so there may be nontrivial upper \emph{and}
lower bounds. To see this, let $y^{\prime}=y$ (so that $\lambda_{1y}=\lambda_{1y^{\prime}}$)
and assume $X_{i1}\beta\geq x\beta$. Then $y^{\prime}\in\mathcal{Y}_{U}$,
so that the observed probability $P\left(\left.Y_{i1}\geq y^{\prime}\right|X_{i}\right)$
provides a nontrivial upper bound:
\[
P\left(\left.Y_{i1}\geq y^{\prime}\right|X_{i}\right)=P\left(\left.Y_{i1}\geq y\right|X_{i}\right)\geq\tau_{1,x,y}\left(X_{i}\right).
\]
However, there \emph{may} be other values in $\underline{\mathcal{Y}}$,
call them $y^{\prime},$ for which $y^{\prime}\in\mathcal{Y}_{L}$.
This would happen if, e.g., $\lambda_{1y}$ and $\lambda_{1y^{\prime}}$
are such that $X_{i1}\beta-\lambda_{1y^{\prime}}\leq x\beta-\lambda_{1y}$.
In this case, the observed probability $P\left(\left.Y_{i1}\geq y^{\prime}\right|X_{i}\right)$
is a nontrivial lower bound:
\[
P\left(\left.Y_{i1}\geq y^{\prime}\right|X_{i}\right)\leq\tau_{1,x,y}\left(X_{i}\right).
\]

The bounds given by Theorem \ref{thm:insight_no_exogeneity} are the
best bounds across $\mathcal{Y}_{U}$ and $\mathcal{Y}_{L}$, i.e.
\begin{equation}
P\left(\left.Y_{i1}\left(x\right)\geq y\right|X_{i}\right)\in\left[\text{max}_{y^{\prime}\in\mathcal{Y}_{L}}\left\{ P\left(\left.Y_{i1}\geq y^{\prime}\right|X_{i}\right)\right\} ,\text{min}_{y^{\prime}\in\mathcal{Y}_{U}}\left\{ P\left(\left.Y_{i1}\geq y^{\prime}\right|X_{i}\right)\right\} \right]\cap\left[0,1\right]\label{eq:ordered_choice_no_exogeneity}
\end{equation}
setting $\min\emptyset=-\infty$ and $\max\emptyset=+\infty$ to deal
with the case when all of $\mathcal{Y}$ provides an upper (lower)
bound.

\subsubsection{Theorem 2 bounds}

Fix $s\in\left\{ 1,\dots,T\right\} $. For each $y^{\prime}\in\mathcal{\underline{Y}}$,
compare $X_{is}\beta-\lambda_{sy^{\prime}}$ to $x\beta-\lambda_{1y}$,
and compute the sets:
\begin{align*}
\mathcal{Y}_{sL} & \equiv\left\{ y^{\prime}\in\mathcal{\underline{Y}}:x\beta-\lambda_{1y}\geq X_{is}\beta-\lambda_{sy^{\prime}}\right\} ,\\
\mathcal{Y}_{sU} & \equiv\left\{ y^{\prime}\in\mathcal{\underline{Y}}:x\beta-\lambda_{1y}\leq X_{is}\beta-\lambda_{sy^{\prime}}\right\} .
\end{align*}

The bounds under time-stationary errors are then the intersection
of the bounds in (\ref{eq:ordered_choice_no_exogeneity}) across \textit{all}
time periods:
\begin{equation}
\max_{s}\max_{y^{\prime}\in\mathcal{Y}_{sL}}\left\{ P\left(\left.Y_{is}\geq y^{\prime}\right|X_{i}\right)\right\} \leq P\left(\left.Y_{i1}\left(x\right)\geq y\right|X_{i}\right)\leq\min_{s}\min_{y^{\prime}\in\mathcal{Y}_{sU}}\left\{ P\left(\left.Y_{is}\geq y^{\prime}\right|X_{i}\right)\right\} .\label{eq:best_bounds}
\end{equation}

\subsection{Censored regression\label{subsec:Example:-censored-regression}}

The structural function
\[
Y_{it}=\max\left\{ 0,\alpha_{i}+X_{it}\beta-U_{it}\right\} 
\]
corresponds to a censored regression model.\footnote{For ease of exposition, we do not consider here extensions covered
by our setup such as $Y_{it}=\max\left\{ \lambda_{t},g_{t}\left(\alpha_{i}+X_{it}\beta-U_{it}\right)\right\} $
with time-varying censoring cutoff, and unknown time-varying $g$
function. For the identification and estimation of the common parameters
in censored regression models, see e.g. \citet{Honore1992,honorePairwiseDifferenceEstimators1994,honoreEstimationTobittypeModels2000a,charlierEstimationCensoredRegression2000,abrevayaIntervalCensoredRegression2020}.
For the cross-sectional case, see e.g. \citet{powellLeastAbsoluteDeviations1984,honorePairwiseDifferenceEstimators1994}.} For this model, \citet{honore_censored_marginal} shows that one
can point-identify a meaningful marginal effect using knowledge of
$\beta$. Because this model is encompassed by our framework, with
$\mathcal{Y}=\left[0,\infty\right)$ and 
\[
h_{t}^{-}\left(y\right)=h^{-}\left(y\right)=\begin{cases}
-\infty & \text{ if }y=0\\
y & \text{ else,}
\end{cases}
\]
we can use our Theorems 1 and 2 to generate additional point and partial
identification results for partial effects in this model.

The object of interest is
\[
\tau_{1,x,y}\left(X_{i}\right)=P\left(\left.Y_{i1}\left(x\right)\geq y\right|X_{i}\right).
\]

The case $y=0$ is not informative because $P\left(\left.Y_{i1}\left(x\right)\geq0\right|X_{i}\right)=1$.
Thus, we restrict attention to $y>0$ and use that $h_{1}^{-}\left(y\right)=y$.
If there exists a $y^{\prime}>0$ such that 
\[
X_{i1}\beta-y^{\prime}=x\beta-y,
\]
i.e. if 
\[
y^{\prime}\equiv y+\left(X_{i1}-x\right)\beta>0,
\]
then
\begin{align*}
Y_{i1}\geq y^{\prime}\Leftrightarrow & \alpha_{i}+X_{i1}\beta-U_{i1}\geq y^{\prime}\\
\Leftrightarrow & \alpha_{i}+X_{i1}\beta-U_{i1}\geq y+\left(X_{i1}-x\right)\beta\\
\Leftrightarrow & \alpha_{i}+x\beta-U_{i1}\geq y\\
\Leftrightarrow & Y_{i1}\left(x\right)\geq y
\end{align*}
so that Theorem \ref{thm:insight_no_exogeneity} implies point identification
\begin{align*}
P\left(\left.Y_{i1}\geq y^{\prime}>0\right|X_{i}\right) & =P\left(\left.Y_{i1}\left(x\right)\geq y\right|X_{i}\right).
\end{align*}

If $y+\left(X_{i1}-x\right)\beta\leq0$ then 
\[
X_{i1}\beta-h_{1}^{-}\left(0\right)\geq x\beta-h_{1}^{-}\left(y\right)
\]
because $h_{1}^{-}\left(0\right)=-\infty$, which yields the trivial
upper bound
\[
P\left(\left.Y_{i1}\left(x\right)\geq y\right|X_{i}\right)\leq1.
\]
Each $y^{\prime}>0$ provides a lower bound, since
\begin{align*}
y+\left(X_{i1}-x\right)\beta & \leq0<y^{\prime}.
\end{align*}
Then 
\[
X_{i1}\beta-h_{1}^{-}\left(y^{\prime}\right)=X_{i1}\beta-y^{\prime}\leq x\beta-y=x\beta-h_{1}^{-}\left(y\right)
\]
so that 
\[
P\left(\left.Y_{i1}\geq y^{\prime}\right|X_{i}\right)\leq P\left(\left.Y_{i1}\left(x\right)\geq y\right|X_{i}\right).
\]
Hence, Theorem \ref{thm:intersection_bound_strict_exogeneity} obtains
\[
\sup_{y^{\prime}>0}P\left(\left.Y_{i1}\geq y^{\prime}\right|X_{i}\right)\leq P\left(\left.Y_{i1}\left(x\right)\geq y\right|X_{i}\right)\leq1.
\]

With time-stationary errors, point identification occurs if there
exists a time period $t$ and a $y^{\prime}>0$ such that 
\[
X_{it}\beta-y^{\prime}=x\beta-y,
\]
because then 
\begin{align*}
P\left(\left.Y_{it}\geq y^{\prime}>0\right|X_{i}\right) & =P\left(\left.Y_{i1}\left(x\right)\geq y>0\right|X_{i}\right).
\end{align*}
This only requires the existence of one time period for which we can
find such a $y^{\prime}>0$.

Partial identification thus only results if, for \emph{each }time
period, 
\[
y+\left(X_{it}-x\right)\beta\leq0,
\]
in which case the resulting bound is
\[
\max_{t}\sup_{y^{\prime}>0}P\left(\left.Y_{it}\geq y^{\prime}\right|X_{i}\right)\leq P\left(\left.Y_{i1}\left(x\right)\geq y\right|X_{i}\right)\leq1.
\]
This partial identification result only applies when the subpopulation
$X_{i}$ is such that for each $t$, $y+\left(X_{it}-x\right)\beta\leq0$.

\section{Numerical experiments\label{subsec:Numerical-example}}

We report on two numerical experiments that explore the bounds in
Theorem 1 and Theorem \ref{thm:intersection_bound_strict_exogeneity}
for the two-way binary choice and ordered choice models. We show that
our bounds are informative without exogeneity or time-homogeneity
assumptions on the error terms. The bounds tighten as $T$ grows,
and they are more informative for ordered choice models than for binary
choice since the bounds tighten as the cardinality of $\underline{\mathcal{Y}}$
grows. Section \ref{subsec:Two-way-binary-choice-contX} presents
results for a two-way binary choice probit model with continuous regressors.
Section \ref{subsec:Staggered-adoption-with} presents results for
a staggered adoption design with both binary and ordered outcomes.

\subsection{Two-way binary choice probit with a continuous regressor\label{subsec:Two-way-binary-choice-contX}}

Consider the following data generating process for a binary choice
model with two-way fixed effects:
\begin{align*}
Y_{it} & =1\left\{ \alpha_{i}+X_{it}\times1-U_{it}\geq\lambda_{t}\right\} ,\\
X_{i1} & \sim\mathcal{N}\left(0,\sigma_{x}^{2}\right),\\
\left.X_{it}\right|X_{it-1} & \sim\mathcal{N}\left(\rho X_{it-1},\sigma_{x}^{2}\right),\\
\left.\alpha_{i}\right|X_{i} & \sim\mathcal{N}\left(X_{i1},1\right),\\
U_{it} & \sim\mathcal{N}\left(0,1\right),\\
\lambda_{t} & =\left(-1+2\frac{\left(t+1\right)}{T}\right)^{2}
\end{align*}
with $\sigma_{x}^{2}=1$ and $\rho=\frac{1}{2}$. The former parameter
controls the cross-sectional heterogeneity in $X_{i}$, while the
latter controls the degree of variation in the sequence $X_{i}$.
The smaller each parameter is, the tighter the bounds are expected
to be. 

Figure \ref{probit_oxford} plots the bounds on the sequence $\mathbb{E}\left[Y_{it}\left(0\right)\right]=\mathbb{E}_{X}\left[\tau_{t,0,1}\left(X_{i}\right)\right]$,
$t=1,2,\dots,T$. The bounds are computed according to Theorem 2.
We find that the bounds get tighter as the number of periods increases.
For example, the width of the interval is $0.26$ when computed with
up to 5 periods, $0.07$ when using all $T=20$ periods. In a separate
experiment with $T=100$ (not reported), the width shrinks to $0.01$
when using all periods. The identified region may not collapse to
a point as $T\rightarrow\infty$, since its width depends on the distribution
of $X_{i}$, the stationary distribution of $\left.\alpha_{i}-U_{it}\right|X_{i}$,
and the shape of $h_{t}$.
\begin{figure}
\includegraphics[scale=0.7]{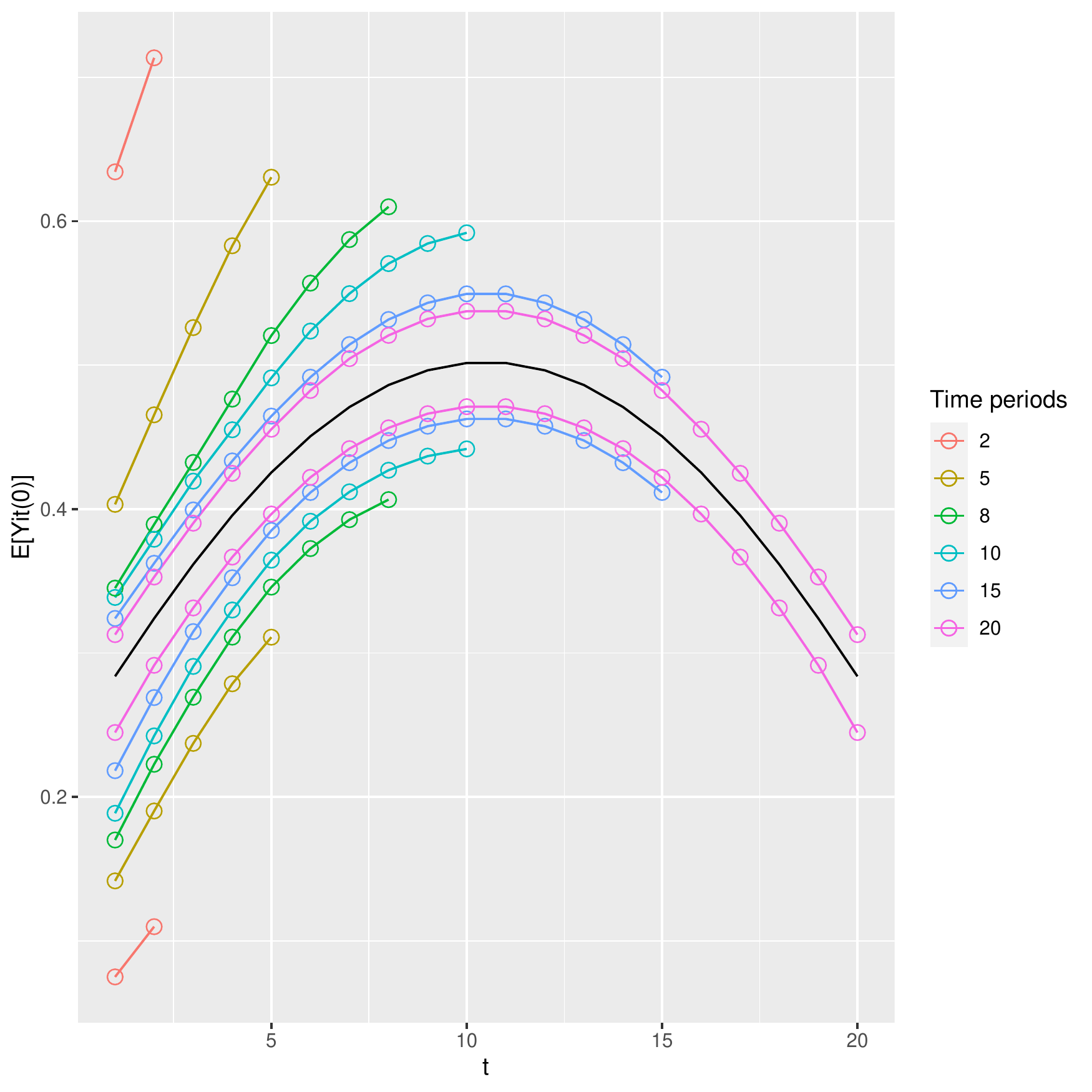}

\caption{Bounds for the two-way binary choice probit model in Section \ref{subsec:Two-way-binary-choice-contX}
with varying number of time periods $T\protect\leq20$. The horizontal
axis indicates the number of time periods $T$. The black solid curve
indicates the true value of $\mathbb{E}\left[Y_{it}\left(0\right)\right]$.
The colored curves are the upper and lower bounds from Theorem \ref{thm:intersection_bound_strict_exogeneity}
using the first $t$ time periods, where $t$ is labeled in the legend.}

\label{probit_oxford}
\end{figure}

Figure \ref{probit_oxford-2} presents results for the sequence $\mathbb{E}\left[Y_{it}\left(0\right)\right],t\geq1$
for some variations on the model above. The first row, left column
shows results for $T=8$, keeping the other design parameters unchanged.
Note that, because $\lambda_{t}=\left(-1+2\frac{\left(t+1\right)}{T}\right)^{2}$,
this changes the evolution of the expectation. All subpanels of Figure
\ref{probit_oxford-2} display results for a deviation from the top
left panel. All results in Figure \ref{probit_oxford-2} are for the
sequence $\mathbb{E}\left[Y_{it}\left(0\right)\right],t\geq1$, except
for the top right panel, which shows the bounds for the sequence $\mathbb{E}\left[Y_{it}\left(1\right)\right],t\geq1$,
under the same model as that for the top left panel. The bounds are
wider because $X_{it}$ has less mass around $1$ than around $0$.
The second row reports bounds on the sequence $\mathbb{E}\left[Y_{it}\left(0\right)\right],t\geq1$,
when there are no time-effects (left column) and when there is no
persistence in $X_{it}$, i.e. $\rho=0$ (right column). The bounds
are slightly wider when the transformation function is time-invariant,
and they are tighter when there is no persistence in $X_{i}$. The
third row shows the bounds when $\sigma_{x}^{2}$ is $1/4$ (left
column; compare to $\sigma_{x}^{2}=1$ in the benchmark case) and
$\sigma_{x}^{2}=4$ (right column). The bounds when there is smaller
cross-sectional variation in $X_{it}$ at a given time period $t$
are tighter than when there is greater cross-sectional variation.
\begin{figure}
\includegraphics[scale=0.7]{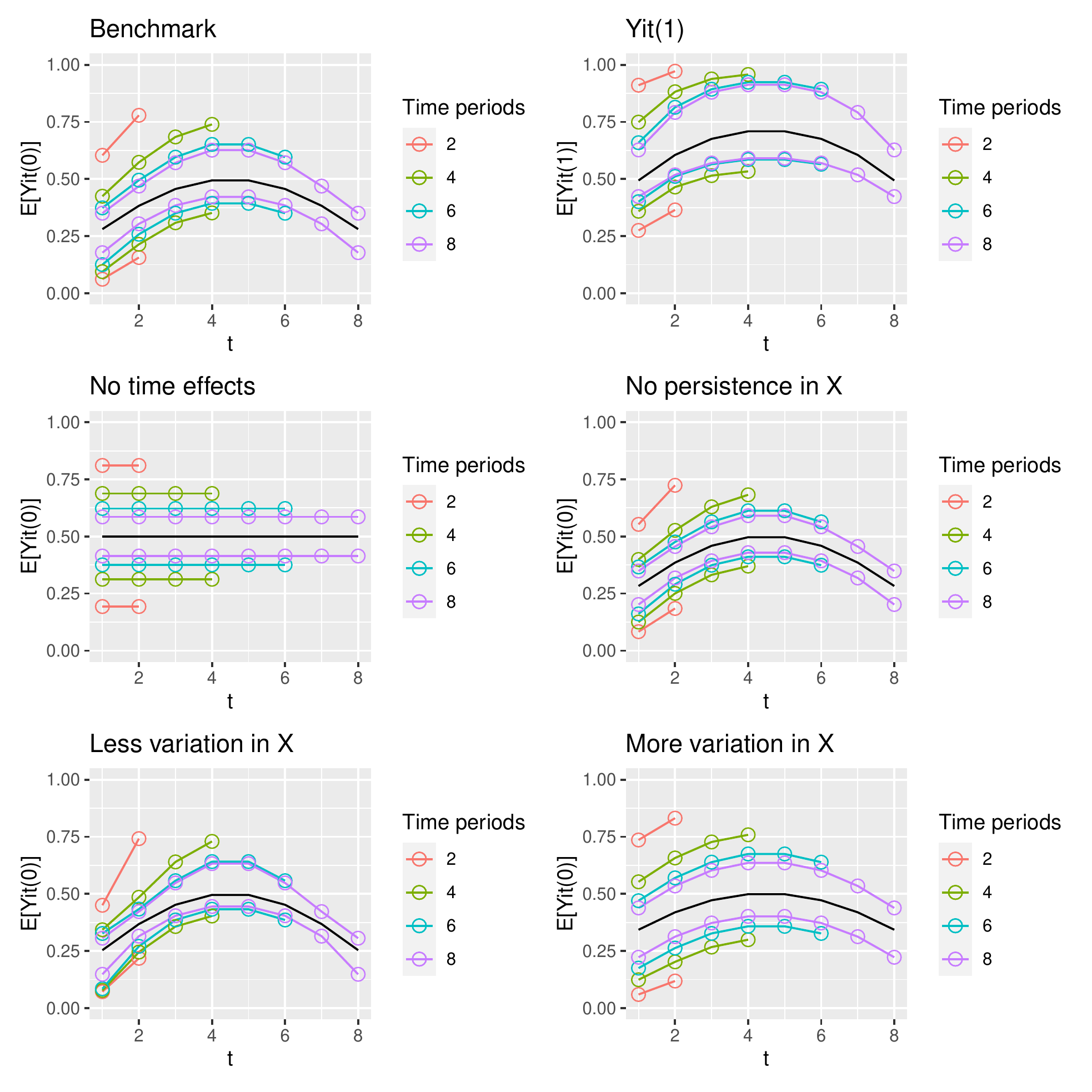}

\caption{Bounds for the two-way binary choice probit model in Section \ref{subsec:Two-way-binary-choice-contX},
and variations, with varying number of time periods $T\protect\leq8$.
All bounds are for $E\left[Y_{it}\left(0\right)\right]$, except for
the top right panel, which shows bounds for $E\left[Y_{it}\left(1\right)\right]$.
The black solid curve indicates the true value of the expectation,
and the colored curves are the upper and lower bounds from Theorem
\ref{thm:intersection_bound_strict_exogeneity} using the first $t$
time periods, where $t$ is labeled in the legend. See the main text
for a description of the model for each panel.}

\label{probit_oxford-2}
\end{figure}

For the specification consider in this section -- binary probit with
two way fixed-effects and short $T$ -- there are no other results
in the literature. In Appendix \ref{CFHN_experiment}, we present
another numerical exercise for binary probit with discrete regressors.
That DGP is the same as the one in Section 8 of \citet{chernozhukov2013average}
when there are no time-effects, i.e. $\lambda_{t}=0$ for all $t$.
In that case, we recover the bounds in \citet{chernozhukov2013average}.

\subsection{Staggered adoption with binary and ordered outcomes\label{subsec:Staggered-adoption-with}}

In this section, we consider a staggered adoption design with both
binary and $J$ ordered outcomes.

The population consists of $G$ groups and individuals in each group
$g\in\left\{ 1,\cdots,G\right\} $ are observed over $T$ periods.
Individuals are untreated up to and including period $g$; they are
treated at period $g+1$, and then stay treated for $t>g+1$, i.e.
\[
X_{it}=\begin{cases}
0 & \text{ if }t\leq g\left(i\right),\\
1 & \text{ if }t>g\left(i\right),
\end{cases}
\]
where $g\left(i\right)$ is individual $i$'s group. Here, $G=19$
and $T\in\left\{ 5,\dots,20\right\} $.

The latent outcome is given by 
\begin{align*}
Y_{it}^{*} & =\alpha_{i}+X_{it}\times1-U_{it},\\
\alpha_{i} & \sim\mathcal{N}\left(0,1\right)+\frac{g\left(i\right)}{G}-\frac{1}{2},
\end{align*}
and $U_{it}$ is standard logistic. The observed outcome is generated
as 
\[
Y_{it}\geq j\Leftrightarrow Y_{it}^{*}\geq\lambda_{jt},
\]
where $j\in J\in\left\{ 2,4,6,8\right\} $\footnote{For $J=2$, the outcome is binary, while for $J>2$, the outcome is
ordered.}, and the threshold $\lambda_{jt}$ is generated as follows. Set $\overline{J}\equiv\frac{J}{2}+1\in\left\{ 1,3,4,5\right\} $,
so that $\left\{ 1,\cdots,\overline{J}-1\right\} $ are the lowest
$J/2$ outcomes, and $\left\{ \overline{J},\cdots,J\right\} $ are
the top $J/2$ outcomes. Set $\lambda_{\overline{J}1}=0,\,\lambda_{\overline{J}t}\sim\mathcal{U}\left[-1,1\right]$,
then draw $\eta_{+},\eta_{-}\sim\mathcal{U}\left[0,1\right]$ and
construct 
\[
\lambda_{jt}=\begin{cases}
\lambda_{\overline{J}t}-\frac{\left(\overline{J}-j\right)}{\overline{J}-1}\eta_{-}, & \text{ if }j<\overline{J},\\
\lambda_{\overline{J}t}+\frac{\left(j-\overline{J}\right)}{\overline{J}-1}\eta_{+}, & \text{ if }j>\overline{J},
\end{cases}
\]
We draw one set of $\left(\left(\lambda_{\overline{J}t}\right)_{t=2}^{T},\eta_{+},\eta_{-}\right)$
and condition our results on them. This allows us to compare the same
parameter across different values of $J$ and $T$.

The parameter of interest is
\[
\tau\equiv\mathbb{E}\left[\tau_{1,1,\overline{J}}\left(X_{i}\right)\right]=\mathbb{E}\left[P\left(\left.Y_{i1}\left(1\right)\geq\overline{J}\right|X_{i}\right)\right]=P\left(Y_{i1}\left(1\right)\geq\overline{J}\right),\,\overline{J}\in\left\{ 1,3,4,5\right\} .
\]
This parameter gives the probability of being in the upper half of
possible outcomes.\footnote{A two-period version of this setup resembles the nonlinear difference-in-difference
setup in \citet{athey2006identification}. The results in this section
differ from theirs because we allow for the combination of fixed effects
and discrete outcomes, which is not covered by their results.}

Our results for this design are presented in Figure \ref{fig:staggered_adoption}.
The solid black line is $\tau=P\left(Y_{i1}\left(1\right)\geq\overline{J}\right)$.
Bounds for different values of $J$ are in colored, dashed lines.
Binary choice is in red. At $T=5$, the width of the bounds for $\tau$
are approximately $0.35$, and become as narrow as $0.05$ when $T=20$.
As $J$ increases, the bounds narrow substantially. This is especially
evident when $T$ is small. For example, at $T=5$, the width of the
bounds for $J=6$ is $0.07$, and for $J=8$ it is $0.02$.
\begin{figure}
\includegraphics[scale=0.7]{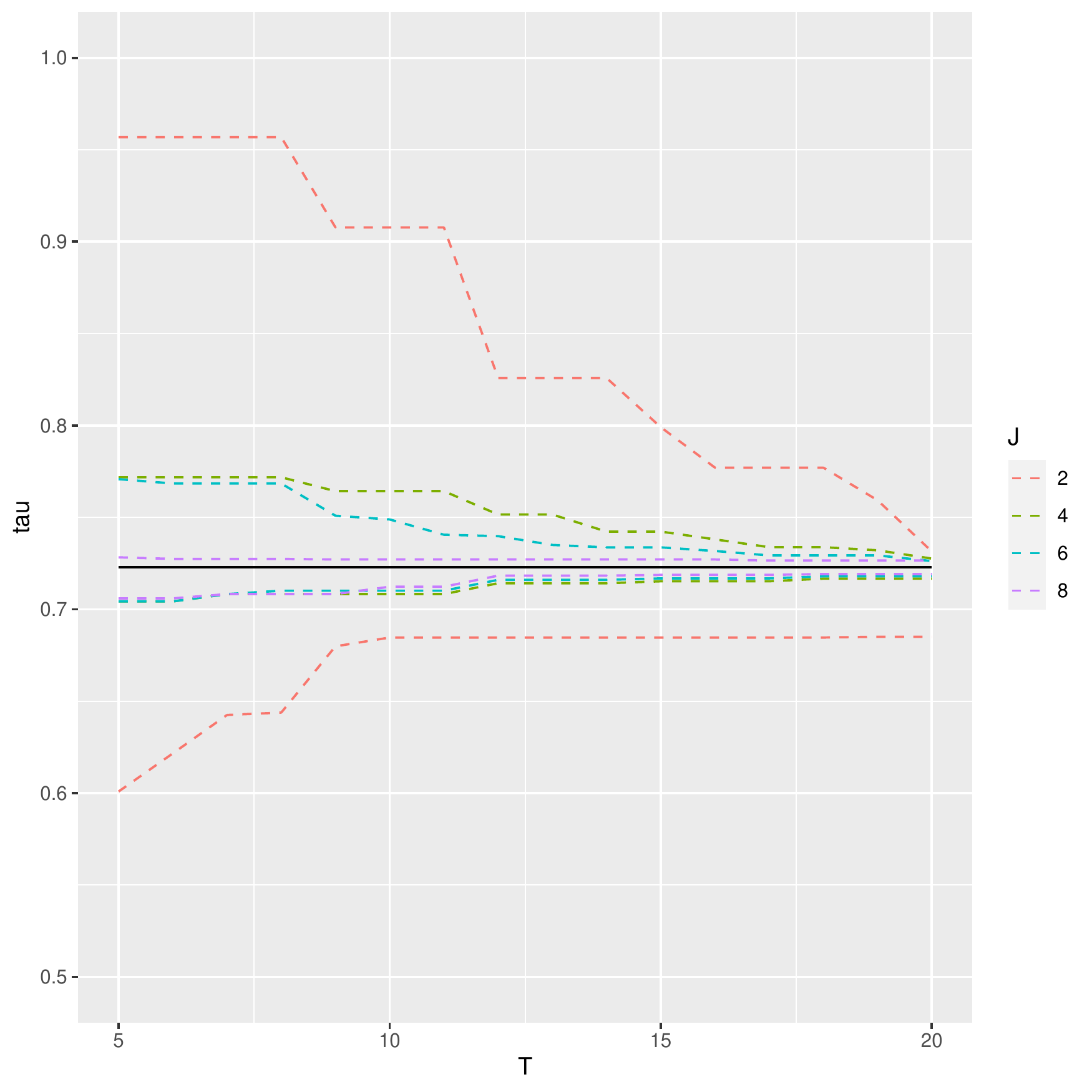}

\caption{Results for staggered adoption with ordered outcomes. The solid black
line is the true value of $\tau\equiv\mathbb{E}\left[\tau_{1,1,\overline{J}}\left(X_{i}\right)\right]$
(vertical axis). Each pair of dashed lines shows the bounds (vertical
axis) as a function of the number of time periods used to construct
the bound (horizontal axis). Colors indicate the number of support
points for the ordered outcome variables, so that red corresponds
to binary choice ($J=2$).}

\label{fig:staggered_adoption}
\end{figure}

Figure \ref{fig:staggered_adoption-2} provides further insight, and
further clarifies the construction of the bounds in Theorem \ref{thm:intersection_bound_strict_exogeneity}.
Consider $J=4$, $G=6$ and $T=8$. Each panel corresponds to a group
$g\left(i\right)$, so that panel ``3'' corresponds to the population
of individuals treated in periods $4-8$. The vertical error bars,
at each $t$, correspond to the best bounds that can be constructed
for the period-$1$ counterfactual using period-$t$ data, see Remark
\ref{rem:alternative_strict_exo_bound_formulation}: 
\[
\left[\sup_{y^{\prime}\in\mathcal{Y}_{tL}}P\left(\left.Y_{is}\geq y^{\prime}\right|X_{i}\right),\inf_{y^{\prime}\in\mathcal{Y}_{tU}}P\left(\left.Y_{is}\geq y^{\prime}\right|X_{i}\right)\right]\cap\left[0,1\right],
\]
with 
\begin{align*}
\mathcal{Y}_{tL} & \equiv\left\{ y^{\prime}\in\mathcal{\underline{\mathcal{Y}}}:X_{is}\beta-h_{s}^{-}\left(y^{\prime}\right)\leq x\beta-h_{t}^{-}\left(y\right)\right\} ,\\
\mathcal{Y}_{tU} & \equiv\left\{ y^{\prime}\in\underline{\mathcal{Y}}:X_{is}\beta-h_{s}^{-}\left(y^{\prime}\right)\geq x\beta-h_{t}^{-}\left(y\right)\right\} .
\end{align*}
From Figure \ref{fig:staggered_adoption-2}, it is clear that $J=4$
is not sufficient to provide non-trivial bounds in each period, see
for example periods 1-3. The bounds in period 4 are also trivial,
but in a way that complements the period 1-3 bounds. Furthermore,
the thresholds in periods 6-8 are such that these periods supply non-trivial
bounds. By using Assumption \ref{assu:Strict-exogeneity}, and taking
the best bounds in each panel over the time periods as in Theorem
\ref{thm:intersection_bound_strict_exogeneity}, relatively narrow
bounds are obtained on $\tau$, see Figure \ref{fig:staggered_adoption}.
\begin{figure}
\includegraphics[scale=0.7]{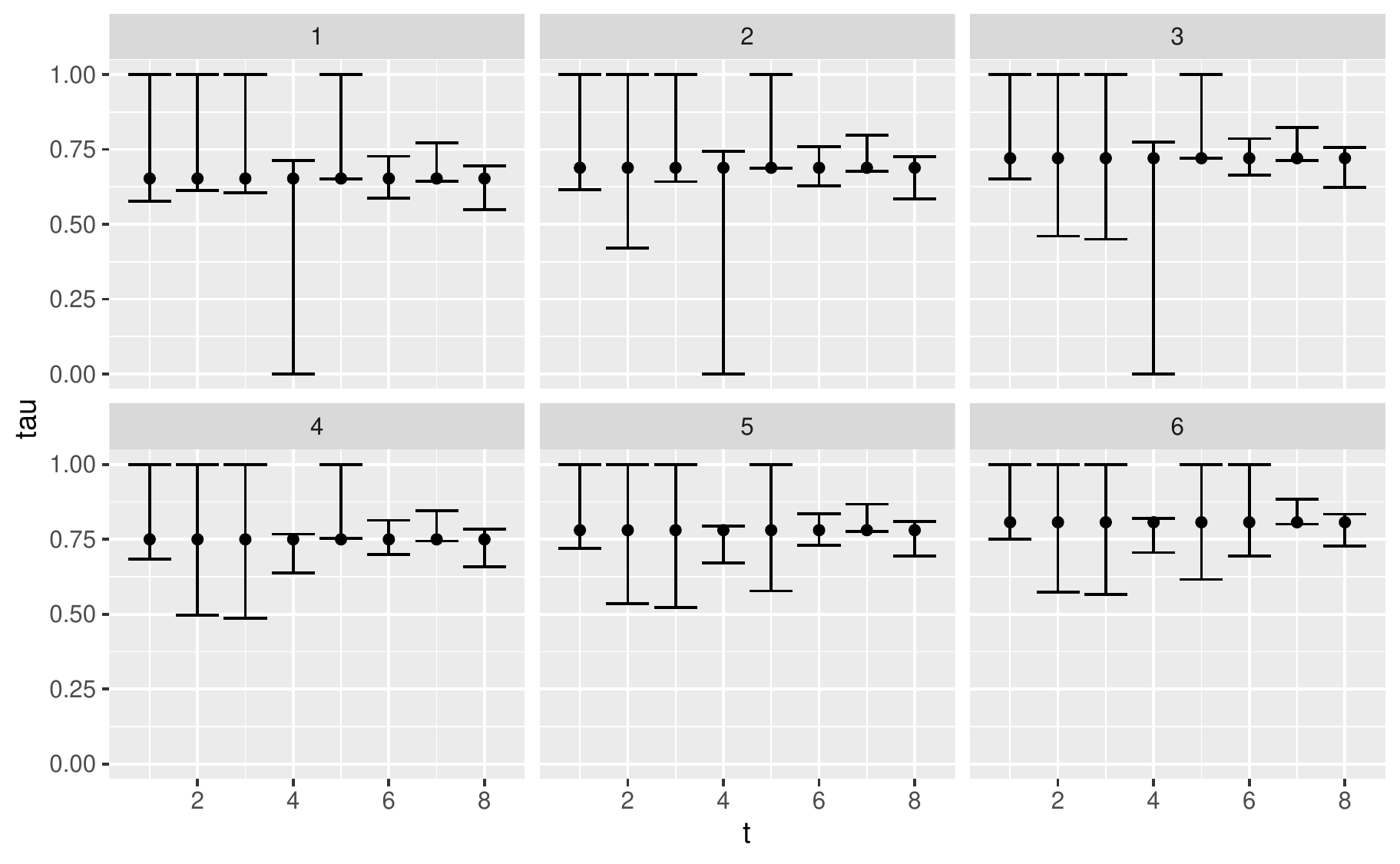}

\caption{Details for the staggered adoption results with $J=4$ and $T=8$.
The panels are numbered $g\in\left\{ 1,\cdots,6\right\} $. For each
panel, $t$ on the horizontal axis is the time period from which the
bound (error bar) is constructed for $P\left(\left.Y_{i1}\left(1\right)\protect\geq3\right|g\left(i\right)=g\right)$.}

\label{fig:staggered_adoption-2}
\end{figure}

\section{Conclusion}

This paper discusses identification of counterfactual parameters in
a class of nonlinear semiparametric panel models with fixed effects
and arbitrary time effects. We derive bounds on the counterfactual
survival probability and related functionals that depend on either
outcomes from the same time period as the counterfactual (without
exogeneity assumptions) or outcomes from across all available time
periods (under a ``strict exogeneity'' or conditional time-stationarity
assumption on the errors). The bounds tighten as the cardinality of
the support of the dependent variable increases, and, under a time-stationarity
assumption on the errors, as $T$ increases. The bounds need not collapse
to a point as $T$ grows, rather they collapse to a point under particular
assumptions on the time effects and the observed regressors, i.e.
when the counterfactual index equals one of the observed indexes.
Our bounds are valid for continuous and discrete outcomes and covariates,
for both movers and stayers.

Although our focus is on identification, we describe here a potential
way of addressing the issue of inference. Note that the bounds in
Theorem \ref{thm:intersection_bound_strict_exogeneity} can be written
as a set of conditional moment inequalities: for fixed $t,x,X_{i},y\in\underline{\mathcal{Y}}$
and for all $\left(s,y^{\prime}\in\underline{\mathcal{Y}}\right)$:
\[
\mathbb{E}\left[\left.\left(X_{is}\beta-h_{s}^{-}\left(y^{\prime}\right)-\left(x\beta-h_{t}^{-}\left(y\right)\right)\right)\times\left(1\left\{ Y_{is}\geq y^{\prime}\right\} -\tau_{t,x,y}\left(X_{i}\right)\right)\right|X_{i}\right]\geq0.
\]
When $\left(\beta,h_{t}\right)$ are partially identified via a set
of moment inequalities, these moment inequalities can be added to
the program. For either pointwise or uniform in $x$ inference, one
could implement the full vector approach of \citet{CoxShi} if the
regressors have finite support, or \citet{AndrewsShi} if the regressors
are continuously distributed.

The bounds of Theorem \ref{thm:insight_no_exogeneity} are valid for
dynamic panel models. However, the bounds may be wider than those
obtained by studying specific models, where the dynamic structure
is known, e.g., binary choice with lagged outcomes as covariates.
In ongoing work, we are studying how to extend our approach to such
models. It is an open question as to whether our approach could be
extended to nonseparable models.

\bibliographystyle{chicago}
\bibliography{new_refs}

\begin{thebibliography}{}

\bibitem[\protect\citeauthoryear{Abrevaya}{Abrevaya}{1999}]{abrevaya1999leapfrog}
Abrevaya, J. (1999).
\newblock Leapfrog {Estimation} of a {Fixed}-{Effects} {Model} with {Unknown} {Transformation} of the {Dependent} {Variable}.
\newblock {\em Journal of Econometrics\/}~{\em 93}, 203--228.

\bibitem[\protect\citeauthoryear{Abrevaya}{Abrevaya}{2000}]{Abrevaya2000}
Abrevaya, J. (2000).
\newblock {Rank Estimation of a Generalized Fixed-Effects Regression Model}.
\newblock {\em Journal of Econometrics\/}~{\em 95\/}(1), 1--23.

\bibitem[\protect\citeauthoryear{Abrevaya and Muris}{Abrevaya and Muris}{2020}]{abrevayaIntervalCensoredRegression2020}
Abrevaya, J. and C.~Muris (2020).
\newblock {Interval Censored Regression with Fixed Effects}.
\newblock {\em Journal of Applied Econometrics\/}~{\em 35\/}(2), 198--216.

\bibitem[\protect\citeauthoryear{Aguirregabiria and Carro}{Aguirregabiria and Carro}{2021}]{aguirregabiriaidentification}
Aguirregabiria, V. and J.~M. Carro (2021).
\newblock Identification of {Average} {Marginal} {Effects} in {Fixed} {Effects} {Dynamic} {Discrete} {Choice} {Models}.
\newblock Technical report.
\newblock Available at https://doi.org/10.48550/arXiv.2107.06141.

\bibitem[\protect\citeauthoryear{Ahn, Ichimura, Powell, and Ruud}{Ahn et~al.}{2018}]{Ahn2018}
Ahn, H., H.~Ichimura, J.~Powell, and P.~Ruud (2018).
\newblock {Simple Estimators for Invertible Index Models}.
\newblock {\em Journal of Business and Economic Statistics\/}~{\em 36\/}(1), 1--10.

\bibitem[\protect\citeauthoryear{Altonji and Matzkin}{Altonji and Matzkin}{2005}]{altonji2005crosssection}
Altonji, J.~G. and R.~L. Matzkin (2005).
\newblock Cross {Section} and {Panel} {Data} {Estimators} for {Nonseparable} {Models} with {Endogenous} {Regressors}.
\newblock {\em Econometrica\/}~{\em 73\/}(4), 1053--1102.

\bibitem[\protect\citeauthoryear{Amemiya and Powell}{Amemiya and Powell}{1981}]{amemiyaComparisonBoxCoxMaximum1981}
Amemiya, T. and J.~L. Powell (1981, December).
\newblock A comparison of the {{Box-Cox}} maximum likelihood estimator and the non-linear two-stage least squares estimator.
\newblock {\em Journal of Econometrics\/}~{\em 17\/}(3), 351--381.

\bibitem[\protect\citeauthoryear{Andrews and Shi}{Andrews and Shi}{2017}]{AndrewsShi}
Andrews, D. and X.~Shi (2017).
\newblock {Inference Based on Many Conditional Moment Inequalities}.
\newblock {\em Journal of Econometrics\/}~{\em 196}, 274--287.

\bibitem[\protect\citeauthoryear{Arellano and Bonhomme}{Arellano and Bonhomme}{2011}]{ArellanoBonhomme2011}
Arellano, M. and S.~Bonhomme (2011).
\newblock {Nonlinear Panel Data Analysis}.
\newblock {\em Annual Review of Economics\/}~{\em 3}, 395--424.

\bibitem[\protect\citeauthoryear{Arellano and Carrasco}{Arellano and Carrasco}{2003}]{ArellanoCarrasco2003}
Arellano, M. and R.~Carrasco (2003).
\newblock {Binary Choice Panel Data Models with Predetermined Variables}.
\newblock {\em Journal of Econometrics\/}~{\em 115}, 125--157.

\bibitem[\protect\citeauthoryear{Arellano and Honor{\'e}}{Arellano and Honor{\'e}}{2001}]{ArellanoHonore2001}
Arellano, M. and Honor{\'e} (2001).
\newblock {Panel Data Models: Some Recent Developments}.
\newblock {\em Handbook of econometrics, Volume 5\/}, 3229--3296.

\bibitem[\protect\citeauthoryear{Aristodemou}{Aristodemou}{2021}]{aristodemou2021semiparametric}
Aristodemou, E. (2021).
\newblock {Semiparametric Identification in Panel Data Discrete Response Models}.
\newblock {\em Journal of Econometrics\/}~{\em 220\/}(2), 253--271.

\bibitem[\protect\citeauthoryear{Athey and Imbens}{Athey and Imbens}{2006}]{athey2006identification}
Athey, S. and G.~W. Imbens (2006).
\newblock Identification and {Inference} in {Nonlinear} {Difference}-in-{Differences} {Model}.
\newblock {\em Econometrica\/}~{\em 74\/}(2), 431--497.

\bibitem[\protect\citeauthoryear{Baetschmann}{Baetschmann}{2012a}]{baetschmannIdentificationEstimationThresholds2012}
Baetschmann, G. (2012a).
\newblock {Identification and Estimation of Thresholds in the Fixed Effects Ordered Logit Model}.
\newblock {\em Economics Letters\/}~{\em 115\/}(3), 416--418.

\bibitem[\protect\citeauthoryear{Baetschmann}{Baetschmann}{2012b}]{baetschmann2012identification}
Baetschmann, G. (2012b).
\newblock {Identification and Estimation of Thresholds in the Fixed Effects Ordered Logit Model}.
\newblock {\em Economics Letters\/}~{\em 115\/}(3), 416--418.

\bibitem[\protect\citeauthoryear{Bartolucci, Pigini, and Valentini}{Bartolucci et~al.}{2023}]{bartolucci_partial_effect2023}
Bartolucci, F., C.~Pigini, and F.~Valentini (2023).
\newblock Conditional inference and bias reduction for partial effects estimation of fixed-effects logit models.
\newblock {\em Empirical Economics\/}~{\em 64}, 2257--2290.

\bibitem[\protect\citeauthoryear{Bester and Hansen}{Bester and Hansen}{2009}]{bester2009identification}
Bester, C.~A. and C.~B. Hansen (2009).
\newblock Identification of {Marginal} {Effects} in a {Nonparametric} {Correlated} {Random} {Effects} {Model}.
\newblock {\em Journal of Business and Economic Statistics\/}~{\em 27\/}(2), 235--250.

\bibitem[\protect\citeauthoryear{Blundell and Powell}{Blundell and Powell}{2003}]{BlundellPowell2003}
Blundell, R.~W. and J.~L. Powell (2003).
\newblock {Endogeneity in Nonparametric and Semiparametric Regression Models}.
\newblock In M.~Dewatripont, L.~Hansen, and S.~Turnovsky (Eds.), {\em Advances in Economics and Econometrics: Theory and Applications}, Volume Eighth World Congress, Vol. II. Oxford: Cambridge University Press.

\bibitem[\protect\citeauthoryear{Blundell and Powell}{Blundell and Powell}{2004}]{BlundellPowell2004}
Blundell, R.~W. and J.~L. Powell (2004).
\newblock {Endogeneity in Semiparametric Binary Response Models}.
\newblock {\em Review of Economic Studies\/}~{\em 71\/}(3), 655--679.

\bibitem[\protect\citeauthoryear{Bonhomme}{Bonhomme}{2012}]{bonhomme2012functional}
Bonhomme, S. (2012).
\newblock Functional {Differencing}.
\newblock {\em Econometrica\/}~{\em 80\/}(4), 1337--1385.

\bibitem[\protect\citeauthoryear{Botosaru and Muris}{Botosaru and Muris}{2017}]{BotosaruMuris2017}
Botosaru, I. and C.~Muris (2017).
\newblock {Binarization for Panel Models with Fixed Effects}.
\newblock Technical report.
\newblock cemmap Working paper CWP31/17, available at https://www.cemmap.ac.uk/publication/binarization-for-panel-models-with-fixed-effects/.

\bibitem[\protect\citeauthoryear{Botosaru, Muris, and Pendakur}{Botosaru et~al.}{2021}]{BotosaruMurisPendakur2021}
Botosaru, I., C.~Muris, and K.~Pendakur (2021).
\newblock {Identification of Time-Varying Transformation Models with Fixed Effects, with an Application to Unobserved Heterogeneity in Resource Shares}.
\newblock {\em Journal of Econometrics\/}.

\bibitem[\protect\citeauthoryear{Botosaru, Muris, and Sokullu}{Botosaru et~al.}{2022}]{BotosaruMurisSokullu2022}
Botosaru, I., C.~Muris, and S.~Sokullu (2022).
\newblock {Partial Effects in Time-Varying Linear Transformation Panel Models with Endogeneity}.
\newblock Technical report.
\newblock University of Bristol, School of Economics Discussion Paper 22/756.

\bibitem[\protect\citeauthoryear{Chamberlain}{Chamberlain}{1984}]{Chamberlain1984}
Chamberlain, G. (1984).
\newblock {Panel Data}.
\newblock In Z.~Griliches and M.~Intriligator (Eds.), {\em Handbook of Econometrics}, Volume~2, Chapter~22, pp.\  1247--1315. Elsevier.

\bibitem[\protect\citeauthoryear{Charlier, Melenberg, and {van Soest}}{Charlier et~al.}{2000}]{charlierEstimationCensoredRegression2000}
Charlier, E., B.~Melenberg, and A.~{van Soest} (2000).
\newblock {Estimation of a Censored Regression Panel Data Model Using Conditional Moment Restrictions Efficiently}.
\newblock {\em Journal of Econometrics\/}~{\em 95\/}(1), 25--56.

\bibitem[\protect\citeauthoryear{Chen}{Chen}{2010}]{Chen2010}
Chen, S. (2010).
\newblock {An Integrated Maximum Score Estimator for a Generalized Censored Quantile Regression Model}.
\newblock {\em Journal of Econometrics\/}~{\em 155\/}(1), 90--98.

\bibitem[\protect\citeauthoryear{Chen, Dahl, and Khan}{Chen et~al.}{2005}]{ChenDahlKahn}
Chen, S., G.~Dahl, and S.~Khan (2005).
\newblock {Nonparametric Identification and Estimation of a Censored Location-Scale Regression Model}.
\newblock {\em Journal of the American Statistical Association\/}~{\em 100\/}(469), 212--221.

\bibitem[\protect\citeauthoryear{Chen, Khan, and Tang}{Chen et~al.}{2019}]{ChenKhanTang2019}
Chen, S., S.~Khan, and X.~Tang (2019).
\newblock {Exclusion Restrictions in Dynamic Binary Choice Panel Data Models: Comment on Semiparametric Binary Choice Panel Data Models Without Strictly Exogenous Regressors}.
\newblock {\em Econometrica\/}~{\em 87}, 1781--1785.

\bibitem[\protect\citeauthoryear{Chen, Lu, and Wang}{Chen et~al.}{2022}]{ChenLuWang2022}
Chen, S., X.~Lu, and X.~Wang (2022).
\newblock {Nonparametric Estimation of Generalized Transformation Models with Fixed Effects}.
\newblock {\em Econometric Theory\/}, 1--32.

\bibitem[\protect\citeauthoryear{Chen and Wang}{Chen and Wang}{2018}]{ChenWang2018}
Chen, S. and X.~Wang (2018).
\newblock {Semiparametric Estimation of Panel Data Models without Monotonicity or Separability}.
\newblock {\em Journal of Econometrics\/}~{\em 206}, 515--530.

\bibitem[\protect\citeauthoryear{Chernozhukov, Fern{\'a}ndez-Val, Hahn, and Newey}{Chernozhukov et~al.}{2013}]{chernozhukov2013average}
Chernozhukov, V., I.~Fern{\'a}ndez-Val, J.~Hahn, and W.~K. Newey (2013).
\newblock Average and {Quantile} {Effects} in {Nonseparable} {Panel} {Models}.
\newblock {\em Econometrica\/}~{\em 81\/}(2), 535--580.

\bibitem[\protect\citeauthoryear{Chernozhukov, Fern{\'a}ndez-Val, Hoderlein, Holzmann, and Newey}{Chernozhukov et~al.}{2015}]{chernozhukov2015nonparametric}
Chernozhukov, V., I.~Fern{\'a}ndez-Val, S.~Hoderlein, H.~Holzmann, and W.~K. Newey (2015).
\newblock Nonparametric {Identification} in {Panels} {Using} {Quantiles}.
\newblock {\em Journal of Econometrics\/}~{\em 188\/}(2), 378--392.

\bibitem[\protect\citeauthoryear{Cox and Shi}{Cox and Shi}{2022}]{CoxShi}
Cox, G. and X.~Shi (2022).
\newblock {Simple Adaptive Size-Exact Testing for Full-Vector and Subvector Inference in Moment Inequality Models}.
\newblock {\em The Review of Economic Studies\/}.

\bibitem[\protect\citeauthoryear{Das and {van Soest}}{Das and {van Soest}}{1999}]{dasPanelDataModel1999}
Das, M. and A.~{van Soest} (1999).
\newblock {A Panel Data Model for Subjective Information on Household Income Growth}.
\newblock {\em Journal of Economic Behavior \& Organization\/}~{\em 40\/}(4), 409--426.

\bibitem[\protect\citeauthoryear{Davezies, D'Haultfoeuille, and Laage}{Davezies et~al.}{2022}]{davezies2021identification}
Davezies, L., X.~D'Haultfoeuille, and L.~Laage (2022).
\newblock Identification and {Estimation} of {Average} {Marginal} {Effects} in {Fixed} {Effects} {Logit} {Models}.
\newblock Technical report.
\newblock Available at https://doi.org/10.48550/arXiv.2105.00879.

\bibitem[\protect\citeauthoryear{Dobronyi, Gu, and Kim}{Dobronyi et~al.}{2021}]{dobronyi2021identification}
Dobronyi, C., J.~Gu, and K.~i. Kim (2021).
\newblock Identification of {Dynamic} {Panel} {Logit} {Models} with {Fixed} {Effects}.
\newblock Technical report.
\newblock Available at https://doi.org/10.48550/arXiv.2104.04590.

\bibitem[\protect\citeauthoryear{Fern{\'a}ndez-Val}{Fern{\'a}ndez-Val}{2009}]{FernandezVal2009}
Fern{\'a}ndez-Val, I. (2009).
\newblock {Fixed Effects Estimation of Structural Parameters and Marginal Effects in Panel Probit Models}.
\newblock {\em Journal of Econometrics\/}~{\em 150}, 71--85.

\bibitem[\protect\citeauthoryear{Fern{\'a}ndez-Val and Weidner}{Fern{\'a}ndez-Val and Weidner}{2018}]{FernandezValWeidner2018}
Fern{\'a}ndez-Val, I. and M.~Weidner (2018).
\newblock {Fixed Effects Estimation of Large-T Panel Data Models}.
\newblock {\em Annual Review of Economics\/}~{\em 10}, 109--138.

\bibitem[\protect\citeauthoryear{Ghanem}{Ghanem}{2017}]{Ghanem2017}
Ghanem, D. (2017).
\newblock {Testing Identifying Assumptions in Nonseparable Panel Data Models}.
\newblock {\em Journal of Econometrics\/}~{\em 197}, 202--217.

\bibitem[\protect\citeauthoryear{Graham and Powell}{Graham and Powell}{2012}]{GrahamPowell2012}
Graham, B. and J.~Powell (2012).
\newblock {Identification and Estimation of Average Partial Effects in "Irregular" Correlated Random Coefficient Panel Data Models}.
\newblock {\em Econometrica\/}~{\em 80\/}(5), 2105--2152.

\bibitem[\protect\citeauthoryear{Han}{Han}{1987}]{Han1987}
Han, A.~K. (1987).
\newblock {Non-Parametric Analysis of a Generalized Regression Model: The Maximum Rank Correlation Estimator}.
\newblock {\em Journal of Econometrics\/}~{\em 35}, 303--316.

\bibitem[\protect\citeauthoryear{Hoderlein and White}{Hoderlein and White}{2012}]{hoderlein2012nonparametric}
Hoderlein, S. and H.~White (2012).
\newblock Nonparametric {Identification} in {Nonseparable} {Panel} {Data} {Models} with {Generalized} {Fixed} {Effects}.
\newblock {\em Journal of Econometrics\/}~{\em 168\/}(2), 300--314.

\bibitem[\protect\citeauthoryear{Honor{\'e}}{Honor{\'e}}{1992}]{Honore1992}
Honor{\'e}, B. (1992).
\newblock {Trimmed LAD and Least Squares Estimation of Truncated and Censored Regression Models with Fixed Effects}.
\newblock {\em Econometrica\/}~{\em 60\/}(3), 533--565.

\bibitem[\protect\citeauthoryear{Honor{\'e}}{Honor{\'e}}{1993}]{Honore1993}
Honor{\'e}, B. (1993).
\newblock {Orthogonality Conditions for Tobit Models with Fixed Effects and Lagged Dependent Variables}.
\newblock {\em Journal of Econometrics\/}~{\em 59}, 35--61.

\bibitem[\protect\citeauthoryear{Honor{\'e} and Hu}{Honor{\'e} and Hu}{2020}]{HonoreHu2020}
Honor{\'e}, B. and L.~Hu (2020).
\newblock {Selection Without Exclusion}.
\newblock {\em Econometrica\/}~{\em 88}, 1007--1029.

\bibitem[\protect\citeauthoryear{Honor{\'e} and Powell}{Honor{\'e} and Powell}{1994}]{honorePairwiseDifferenceEstimators1994}
Honor{\'e}, B. and J.~Powell (1994).
\newblock Pairwise difference estimators of censored and truncated regression models.
\newblock {\em Journal of Econometrics\/}~{\em 64\/}(1-2), 241--278.

\bibitem[\protect\citeauthoryear{Honor{\'e} and Tamer}{Honor{\'e} and Tamer}{2006}]{honore2006boundson}
Honor{\'e}, B. and E.~Tamer (2006).
\newblock Bounds on {Parameters} in {Panel} {Dynamic} {Discrete} {Choice} {Models}.
\newblock {\em Econometrica\/}~{\em 74\/}(3), 611--629.

\bibitem[\protect\citeauthoryear{Honore}{Honore}{2008}]{honore_censored_marginal}
Honore, B.~E. (2008).
\newblock On marginal effects in semiparametric censored regression models.
\newblock Technical report.
\newblock Available at SSRN: https://ssrn.com/abstract=1394384 or http://dx.doi.org/10.2139/ssrn.1394384.

\bibitem[\protect\citeauthoryear{Honor{\'e}, Kyriazidou, and Powell}{Honor{\'e} et~al.}{2000}]{honoreEstimationTobittypeModels2000a}
Honor{\'e}, B.~E., E.~Kyriazidou, and J.~L. Powell (2000, January).
\newblock Estimation of tobit-type models with individual specific effects.
\newblock {\em Econometric Reviews\/}~{\em 19\/}(3), 341--366.

\bibitem[\protect\citeauthoryear{Horowitz and Lee}{Horowitz and Lee}{2004}]{HorowitzLee2004}
Horowitz, J. and S.~Lee (2004).
\newblock {Semiparametric Estimation of a Panel Data Proportional Hazards Model with Fixed Effects}.
\newblock {\em Journal of Econometrics\/}~{\em 119}, 155--198.

\bibitem[\protect\citeauthoryear{Ichimura}{Ichimura}{1993}]{Ichimura}
Ichimura, H. (1993).
\newblock {Semiparametric Least Squares (SLS) and Weighted SLS Estimation of Single-Index Models}.
\newblock {\em Journal of Econometrics\/}~{\em 58}, 71--120.

\bibitem[\protect\citeauthoryear{Johnson}{Johnson}{2004}]{johnsonphdthesis}
Johnson, E.~G. (2004).
\newblock {\em {Panel Data Models With Discrete Dependent Variables}}.
\newblock Phd thesis, stanford university.

\bibitem[\protect\citeauthoryear{Khan, Ponomareva, and Tamer}{Khan et~al.}{2011}]{KhanPonomarevaTamer2011}
Khan, S., M.~Ponomareva, and E.~Tamer (2011).
\newblock {Sharpness in Randomly Censored Linear Models}.
\newblock {\em Economic Letters\/}~{\em 113}, 23--25.

\bibitem[\protect\citeauthoryear{Khan, Ponomareva, and Tamer}{Khan et~al.}{2016}]{KhanPonomarevaTamer2016}
Khan, S., M.~Ponomareva, and E.~Tamer (2016).
\newblock {Identification of Panel Data Models with Endogenous Censoring}.
\newblock {\em Journal of Econometrics\/}~{\em 94}, 57--75.

\bibitem[\protect\citeauthoryear{Khan, Ponomareva, and Tamer}{Khan et~al.}{2023}]{KhanPonomarevaTamer2021}
Khan, S., M.~Ponomareva, and E.~Tamer (2023).
\newblock {Identification of Dynamic Binary Response Models}.
\newblock {\em Journal of Econometrics\/}~{\em 237}.

\bibitem[\protect\citeauthoryear{Lee}{Lee}{2008}]{Lee2008}
Lee, S. (2008).
\newblock {Estimating Panel Data Duration Models with Censored Data}.
\newblock {\em Econometric Theory\/}~{\em 24}, 1254--1276.

\bibitem[\protect\citeauthoryear{Lin and Wooldridge}{Lin and Wooldridge}{2015}]{LinWooldridge2015}
Lin, W. and J.~M. Wooldridge (2015).
\newblock {On Different Approaches to Obtaining Partial Effects in Binary Response Models with Endogenous Regressors}.
\newblock {\em Economics Letters\/}~{\em 134}, 58--61.

\bibitem[\protect\citeauthoryear{Liu, Poirier, and Shiu}{Liu et~al.}{2023}]{liu2021identification}
Liu, L., A.~Poirier, and J.-L. Shiu (2023).
\newblock Identification and {Estimation} of {Average} {Partial} {Effects} in {Semiparametric} {Binary} {Response} {Panel} {Models}.
\newblock Technical report.
\newblock Available at https://doi.org/10.48550/arXiv.2105.12891.

\bibitem[\protect\citeauthoryear{Manski}{Manski}{1987}]{manski1987semiparametric}
Manski, C. (1987).
\newblock Semiparametric {Analysis} of {Random} {Effects} {Linear} {Models} {From} {Binary} {Panel} {Data}.
\newblock {\em Econometrica\/}~{\em 55\/}(2), 357--362.

\bibitem[\protect\citeauthoryear{Muris}{Muris}{2017}]{muris2017estimation}
Muris, C. (2017).
\newblock Estimation in the {Fixed}-{Effects} {Ordered} {Logit} {Model}.
\newblock {\em Review of Economics and Statistics\/}~{\em 99\/}(3), 465--477.

\bibitem[\protect\citeauthoryear{Pakel and Weidner}{Pakel and Weidner}{2023}]{PakelWeidner}
Pakel, C. and M.~Weidner (2023).
\newblock {Bounds on Average Effects in Discrete Choice Panel Data Models}.
\newblock Technical report.
\newblock Available at https://doi.org/10.48550/arXiv.2309.09299.

\bibitem[\protect\citeauthoryear{Powell}{Powell}{1991}]{barnettNonparametricSemiparametricMethods1991}
Powell, J. (1991).
\newblock {Monotonic Regression Models under Quantile Restrictions}.
\newblock In W.~A. Barnett, J.~Powell, and G.~E. Tauchen (Eds.), {\em Nonparametric and {{Semiparametric Methods}} in {{Econometrics}} and {{Statistics}}: {{Proceedings}} of the {{Fifth International Symposium}} in {{Economic Theory}} and {{Econometrics}}}, Chapter~14, pp.\  357--384. {Cambridge University Press}.

\bibitem[\protect\citeauthoryear{Powell, Stock, and Stoker}{Powell et~al.}{1989}]{PowellStockStoker}
Powell, J., J.~Stock, and T.~Stoker (1989).
\newblock Semiparametric estimation of index coefficients.
\newblock {\em Econometrica\/}~{\em 57\/}(6), 1403--1430.

\bibitem[\protect\citeauthoryear{Powell}{Powell}{1984}]{powellLeastAbsoluteDeviations1984}
Powell, J.~L. (1984, July).
\newblock Least absolute deviations estimation for the censored regression model.
\newblock {\em Journal of Econometrics\/}~{\em 25\/}(3), 303--325.

\bibitem[\protect\citeauthoryear{Powell}{Powell}{1996}]{powellRescaledMethodsofmomentsEstimation1996}
Powell, J.~L. (1996, June).
\newblock Rescaled methods-of-moments estimation for the {{Box-Cox}} regression model.
\newblock {\em Economics Letters\/}~{\em 51\/}(3), 259--265.

\bibitem[\protect\citeauthoryear{Song and Wang}{Song and Wang}{2021}]{SongWang}
Song, P. and S.~Wang (2021).
\newblock {A Further Remark on the Alternative Expectation Formula}.
\newblock {\em Communications in Statistics - Theory and Methods\/}~{\em 50\/}(11), 2586--2591.

\bibitem[\protect\citeauthoryear{Wang and Chen}{Wang and Chen}{2020}]{WangChen2020}
Wang, X. and S.~Chen (2020).
\newblock {Semiparametric Estimation of Generalized Transformation Panel Data Models with Non-Stationary Error}.
\newblock {\em The Econometrics Journal\/}~{\em 23\/}(3), 386--402.

\end{thebibliography}

\appendix

\section{Proofs}
\begin{proof}[Proof of Theorem \ref{thm:insight_no_exogeneity}.]
 Define the (potential) latent variables
\begin{align}
Y_{it}^{*}\left(x\right) & \equiv\alpha_{i}+x\beta-U_{it},\label{eq:latent variable}\\
Y_{it}^{*} & \equiv Y_{it}^{*}\left(X_{it}\right).\nonumber 
\end{align}

Assumption \ref{assu:weak_mono} obtains the following equivalent
relationships:
\begin{align}
 & Y_{it}\left(x\right)\geq y\nonumber \\
 & \Leftrightarrow Y_{it}^{*}\left(x\right)\geq h_{t}^{-}\left(y\right)\\
 & \Leftrightarrow U_{it}-\alpha_{i}\leq x\beta-h_{t}^{-}\left(y\right),\label{eq:monotonicity_implication}
\end{align}
that is, the counterfactual outcome being greater than a fixed value
$y$ is equivalent to the random variable $U_{it}-\alpha_{i}$ being
smaller than the counterfactual index evaluated at $y$.

Then, for any value $y^{\prime}\in\mathcal{Y}$, the following are
equivalent:
\begin{align*}
X_{it}\beta-h_{t}^{-}\left(y^{\prime}\right)\geq & x\beta-h_{t}^{-}\left(y\right)\\
\Leftrightarrow\\
P\left(\left.Y_{it}\geq y^{\prime}\right|X_{i}\right)\geq & P\left(\left.Y_{it}\left(x\right)\geq y\right|X_{i}\right)=\tau_{t,x,y}\left(X_{i}\right).
\end{align*}

To see this, suppose that there exists a $y^{\prime}$ such that 
\[
X_{it}\beta-h_{t}^{-}\left(y^{\prime}\right)\geq x\beta-h_{t}^{-}\left(y\right).
\]
 Then:
\begin{align*}
P\left(\left.Y_{it}\geq y^{\prime}\right|X_{i}\right) & =P\left(\left.Y_{it}^{*}\geq h_{t}^{-}\left(y^{\prime}\right)\right|X_{i}\right)\\
 & =P\left(\left.\alpha_{i}+X_{it}\beta-U_{it}\geq h_{t}^{-}\left(y^{\prime}\right)\right|X_{i}\right)\\
 & =P\left(\left.U_{it}-\alpha_{i}\leq X_{it}\beta-h_{t}^{-}\left(y^{\prime}\right)\right|X_{i}\right)\\
 & \geq P\left(\left.U_{it}-\alpha_{i}\leq x\beta-h_{t}^{-}\left(y\right)\right|X_{i}\right)\\
 & =P\left(\left.\alpha_{i}+x\beta-U_{it}\geq h_{t}^{-}\left(y\right)\right|X_{i}\right)\\
 & =P\left(\left.Y_{it}\left(x\right)\geq y\right|X_{i}\right)\\
 & =\tau_{t,x,y}\left(X_{i}\right),
\end{align*}
where the inequality follows by weak monotonicity of CDFs.
\end{proof}
\begin{proof}[Proof of Theorem \ref{thm:intersection_bound_strict_exogeneity}]
\citealt{chernozhukov2015nonparametric} point out that Assumption
\ref{assu:Strict-exogeneity} is equivalent to $\left.\alpha_{i},U_{it}\right|X_{i}\stackrel{d}{=}\left.\alpha_{i},U_{i1}\right|X_{i}$.
In particular,
\begin{equation}
\left.U_{it}-\alpha_{i}\right|X_{i}\stackrel{d}{=}\left.U_{is}-\alpha_{i}\right|X_{i},\,s,t\in\left\{ 1,\dots,T\right\} .\label{eq:impliedbystat}
\end{equation}

Then, for any $\left(x,X_{i},t,y,s,y^{\prime}\right)$ we have that
\begin{align*}
X_{is}\beta-h_{s}^{-}\left(y^{\prime}\right)\geq & x\beta-h_{t}^{-}\left(y\right)\\
\Leftrightarrow\\
P\left(\left.Y_{is}\geq y^{\prime}\right|X_{i}\right)\geq & P\left(\left.Y_{it}\left(x\right)\geq y\right|X_{i}\right)=\tau_{t,x,y}\left(X_{i}\right).
\end{align*}
This follows from 
\begin{align*}
P\left(\left.Y_{is}\geq y^{\prime}\right|X_{i}\right) & =P\left(\left.Y_{is}^{*}\geq h_{s}^{-}\left(y^{\prime}\right)\right|X_{i}\right)\\
 & =P\left(\left.\alpha_{i}+X_{is}\beta-U_{is}\geq h_{s}^{-}\left(y^{\prime}\right)\right|X_{i}\right)\\
 & =P\left(\left.U_{is}-\alpha_{i}\leq X_{is}\beta-h_{s}^{-}\left(y^{\prime}\right)\right|X_{i}\right)\\
 & =P\left(\left.U_{it}-\alpha_{i}\leq X_{is}\beta-h_{s}^{-}\left(y^{\prime}\right)\right|X_{i}\right)\\
 & \geq P\left(\left.U_{it}-\alpha_{i}\leq x\beta-h_{t}^{-}\left(y\right)\right|X_{i}\right)\\
 & =P\left(\left.\alpha_{i}+x\beta-U_{it}\geq h_{t}^{-}\left(y\right)\right|X_{i}\right)\\
 & =P\left(\left.Y_{it}\left(x\right)\geq y\right|X_{i}\right)\\
 & =\tau_{t,x,y}\left(X_{i}\right).
\end{align*}
The argument is similar to that in the proof of Theorem \ref{thm:insight_no_exogeneity},
except that the fourth equality follows by (\ref{eq:impliedbystat}),
and the fifth inequality follows by weak monotonicity of CDFs.
\end{proof}

\section{Numerical experiment with discrete regressors\label{CFHN_experiment}}

We consider here a numerical exercise inspired by Section 8 in \citet{chernozhukov2013average}:
a probit model with discrete regressors. When the structural equation
is time-invariant, the data generating process corresponds to that
in \citet{chernozhukov2013average}. We present results for, first,
the time-invariant case -- for which our bounds correspond to those
in \citet{chernozhukov2013average}, and then for a time-varying case.

The data generating process is:
\begin{align}
Y_{it} & =1\left\{ \alpha_{i}+X_{it}\beta-U_{it}-\lambda_{t}\geq0\right\} ,\,t=1,2,\label{eq:appB_model}\\
\alpha_{i} & \sim\mathcal{N}\left(0,1\right),\nonumber \\
U_{it} & \sim\mathcal{N}\left(0,1\right),\nonumber \\
\eta_{it} & \sim\mathcal{N}\left(0,1\right),\,t=1,2,\nonumber \\
X_{it} & =1\left\{ \alpha_{i}\geq\eta_{it}\right\} ,\,t=1,2,\nonumber \\
\lambda_{1} & =0,\nonumber \\
\lambda_{2} & \in\left\{ 0,\frac{1}{2},1\right\} ,\nonumber \\
\beta & \in\left[0,2\right].\nonumber 
\end{align}

We consider the ATE at $t=1,2$, defined as:
\begin{align*}
ATE_{t} & \equiv\mathbb{E}_{X}\left(P\left(\left.Y_{it}\left(1\right)\geq1\right|X_{i}\right)-P\left(\left.Y_{it}\left(0\right)\geq1\right|X_{i}\right)\right)\\
 & =\mathbb{E}_{X}\left(\tau_{t,1,1}\left(X_{i}\right)-\tau_{t,0,1}\left(X_{i}\right)\right).
\end{align*}
Recall that $\tau_{t,0,1}\left(X_{i}\right)$ refers to the counterfactual
probability that in period $t$ (first index), the probability that
the counterfactual outcome under $x=0$ (second index) for the the
subpopulation with $X_{i}$ equals or exceeds $y=1$ (third index).
For simplicity, we set $t=1$, so that the counterfactual index for
the analysis that follows is:
\begin{equation}
x\beta-\lambda_{1}=0.\label{eq:counter_idex}
\end{equation}
According to Theorem 2, if the observed index is smaller (greater)
than the counterfactual index (\ref{eq:counter_idex}), the observed
probability associated with the observed index provides a lower (upper)
bound on the counterfactual probability, while if the observed index
equals the counterfactual index, the observed probability point identifies
the counterfactual probability.

First, in order to compare our results to those in \citet{chernozhukov2013average},
we set $\lambda_{2}=0$.\footnote{In this case, there are no time effects, so the outcome equation is
time-homogeneous, as required by the results of \citet{chernozhukov2013average}.} Theorem 2 then obtains the following results for $\tau_{1,0,1}\left(X_{i}\right)$
and $\tau_{1,1,1}\left(X_{i}\right)$. Note that $\tau_{1,0,1}\left(X_{i}\right)$
is point-identified for the subpopulations with 
\[
X_{i}\in\left\{ \left(0,0\right),\left(0,1\right),\left(1,0\right)\right\} ,
\]
because for the listed subpopulations, the value $x=0$ is observed
in at least one of the periods. Thus:
\begin{align*}
\tau_{1,0,1}\left(\left(0,0\right)\right) & =P\left(\left.Y_{i1}\geq1\right|X_{i}=\left(0,0\right)\right)=P\left(\left.Y_{i2}\geq1\right|X_{i}=\left(0,0\right)\right),\\
\tau_{1,0,1}\left(\left(0,1\right)\right) & =P\left(\left.Y_{i1}\geq1\right|X_{i}=\left(0,1\right)\right),\\
\tau_{1,0,1}\left(\left(1,0\right)\right) & =P\left(\left.Y_{i2}\geq1\right|X_{i}=\left(1,0\right)\right).
\end{align*}

For the subpopulation of \emph{stayers} with $X_{i}=\left(1,1\right)$,
the value $x=0$ is not observed in either time period, so the counterfactual
probability $\tau_{1,0,1}\left(\left(1,1\right)\right)$ is partially
identified unless $\beta=0$:
\[
\tau_{1,0,1}\left(\left(1,1\right)\right)\in\begin{cases}
\left[0,P\left(\left.Y_{i1}\geq1\right|X_{i}=\left(1,1\right)\right)\right], & \beta>0,\\
\left[P\left(\left.Y_{i1}\geq1\right|X_{i}=\left(1,1\right)\right),1\right], & \beta<0,\\
\left\{ P\left(\left.Y_{i1}\geq1\right|X_{i}=\left(1,1\right)\right)\right\} . & \beta=0.
\end{cases}
\]
This is because the observed index for this subpopulation is $\beta$
in both time periods.

The analysis for $\tau_{1,1,1}\left(X_{i}\right)$ is similar and
omitted. Figure \ref{probit_CFHN} plots the bounds for $ATE_{t=1}$.
Note that our bounds for this particular example coincide with those
in \citet{chernozhukov2013average} that the authors call ``NPM.''
\begin{figure}
\includegraphics[scale=0.4]{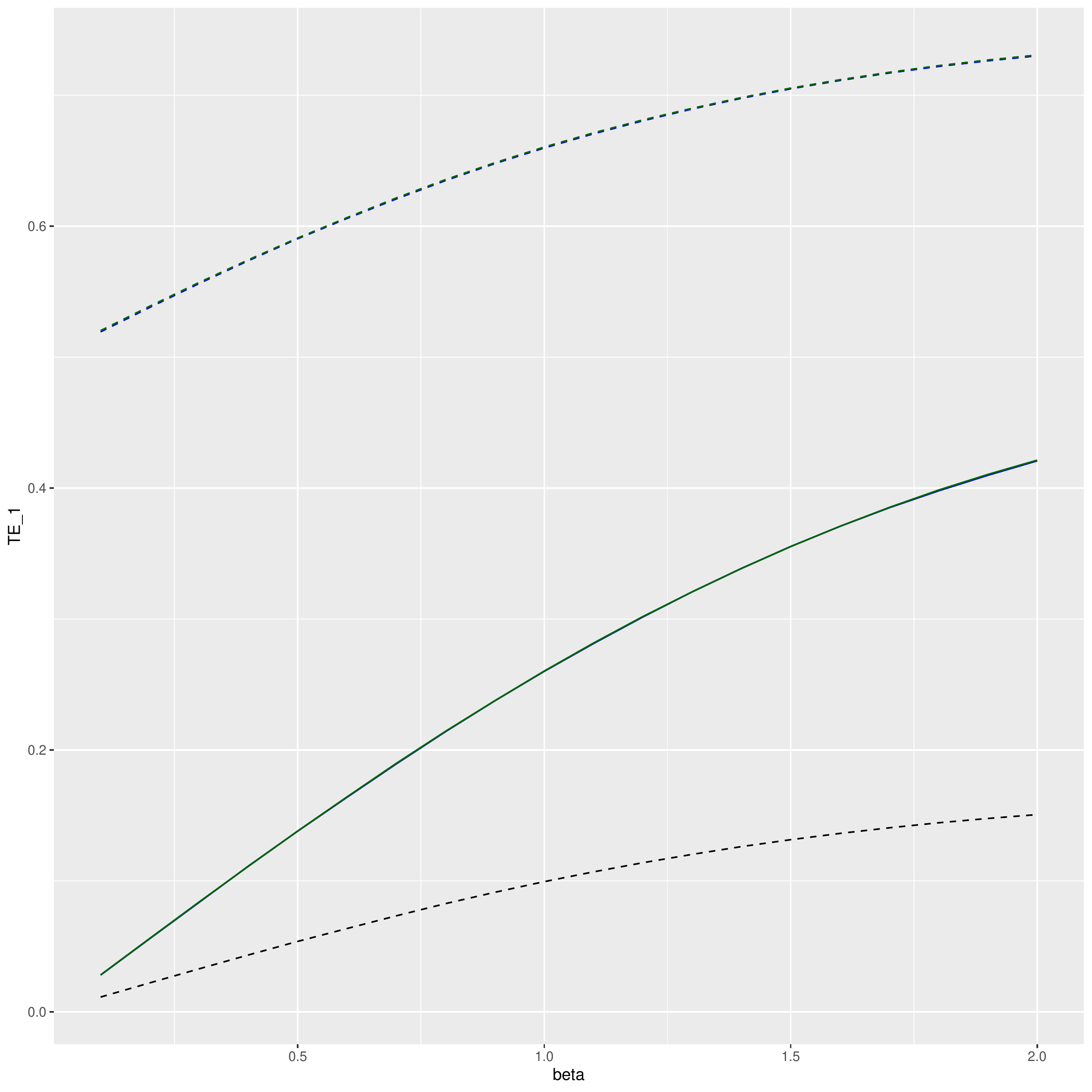}

\caption{Bounds for binary choice probit model with discrete regressors and
no time effects in (\ref{eq:appB_model}). The solid line shows the
true ATE at $t=1$ on the vertical axis, with the true value of $\beta$
on the horizontal axis. The dotted lines show the upper and lower
bounds from Theorem \ref{thm:intersection_bound_strict_exogeneity},
described in the main text of Appendix \ref{CFHN_experiment}.}

\label{probit_CFHN}
\end{figure}

Consider now the specification with positive time effects, $\lambda_{2}>0$.
Theorem 2 then obtains the following results for $\tau_{1,0,1}\left(X_{i}\right)$
and $\tau_{1,1,1}\left(X_{i}\right)$. 

First, $\tau_{1,0,1}\left(X_{i}\right)$ is point-identified for the
subpopulations with $X_{i1}=0$: the stayers with $X_{i}=\left(0,0\right)$
and the movers with $X_{i}=\left(0,1\right)$. This is because $x=0$
at $t=1$ for those subpopulations, so that the observed index for
these subpopulations equals the counterfactual index. Thus, 
\begin{align*}
\tau_{1,0,1}\left(\left(0,0\right)\right) & =P\left(\left.Y_{i1}\geq1\right|X_{i}=\left(0,0\right)\right),\\
\tau_{1,0,1}\left(\left(0,1\right)\right) & =P\left(\left.Y_{i1}\geq1\right|X_{i}=\left(0,1\right)\right).
\end{align*}

Second, for the subpopulation of movers with $X_{i}=\left(1,0\right)$,
the period $t=2$ observed probability provides a lower bound on $\tau_{1,0,1}\left(X_{i}\right)$
because the observed index for this subpopulation is $X_{i2}\beta-\lambda_{2}=-\lambda_{2}$,
which is smaller than the counterfactual index since we assumed that
$\lambda_{2}>0$, so:
\[
\tau_{1,0,1}\left(\left(1,0\right)\right)\geq P\left(\left.Y_{i2}\geq1\right|\left(1,0\right)\right).
\]
Whether the $t=1$ observed probabilities point or partially identify
$\tau_{1,0,1}\left(\left(1,0\right)\right)$ for this subpopulation
depends on how the observed index $X_{i1}\beta-\lambda_{1}=\beta$
compares to the counterfactual index in (\ref{eq:counter_idex}):
\begin{align*}
\tau_{1,0,1}\left(\left(1,0\right)\right) & =P\left(\left.Y_{i1}\geq\right|\left(1,0\right)\right),\text{ if }\beta=0,\\
\tau_{1,0,1}\left(\left(1,0\right)\right) & <P\left(\left.Y_{i1}\geq\right|\left(1,0\right)\right),\text{ if }\beta>0,\\
\tau_{1,0,1}\left(\left(1,0\right)\right) & >P\left(\left.Y_{i1}\geq\right|\left(1,0\right)\right),\text{ if }\beta<0.
\end{align*}

Third, for the subpopulation of stayers with $X_{i}=\left(1,1\right)$,
the observed index at $t=1$ is $X_{i1}\beta-\lambda_{1}=\beta$,
so that comparing it to the counterfactual index in (\ref{eq:counter_idex})
obtains: 
\begin{align*}
\tau_{1,0,1}\left(\left(1,1\right)\right) & =P\left(\left.Y_{i1}\geq\right|\left(1,1\right)\right),\text{ if }\beta=0,\\
\tau_{1,0,1}\left(\left(1,1\right)\right) & <P\left(\left.Y_{i1}\geq\right|\left(1,1\right)\right),\text{ if }\beta>0,\\
\tau_{1,0,1}\left(\left(1,1\right)\right) & >P\left(\left.Y_{i1}\geq\right|\left(1,1\right)\right),\text{ if }\beta<0,
\end{align*}
while the observed index at $t=2$ for this subpopulation is $X_{i2}\beta-\lambda_{2}=\beta-\lambda_{2}$,
which compared to the same counterfactual index in (\ref{eq:counter_idex})
obtains: 
\begin{align*}
\tau_{1,0,1}\left(\left(1,1\right)\right) & =P\left(\left.Y_{i2}\geq\right|\left(1,1\right)\right),\text{ if }\beta=\lambda_{2},\\
\tau_{1,0,1}\left(\left(1,1\right)\right) & >P\left(\left.Y_{i2}\geq\right|\left(1,1\right)\right),\text{ if }\beta<\lambda_{2},\\
\tau_{1,0,1}\left(\left(1,1\right)\right) & <P\left(\left.Y_{i2}\geq\right|\left(1,1\right)\right),\text{ if }\beta>\lambda_{2}.
\end{align*}

We have now described all the restrictions underlying Theorem 2 for
$\tau_{1,0,1}\left(X_{i}\right)$ for all $X_{i}$ in our example.
The bounds now follow by selecting the best bounds for each case and
for each $X_{i}$ as in Theorem 2. The same analysis can be done for
$\tau_{1,1,1}\left(X_{i}\right)$ and for period 2 counterfactuals. 

Figure \ref{probit_CFHN_time} plots the bounds for ATE$_{t=1}$ (left
panel) and ATE$_{t=2}$ (right panel). The solid lines correspond
to the true ATEs, while the dashed lines correspond to their respective
bounds, grouped by color. The colors correspond to different values
of $\lambda_{2}$. For example, in both panels, the red lines correspond
to the true ATE and its bounds when $\lambda_{2}=0$ (no time-effects).
In the left panel, all true ATEs have the same value since there are
no time-effects, $\lambda_{1}=0$, while in the right panel the true
ATEs have different values because they correspond to different time
effects.
\begin{figure}
\includegraphics[scale=0.7]{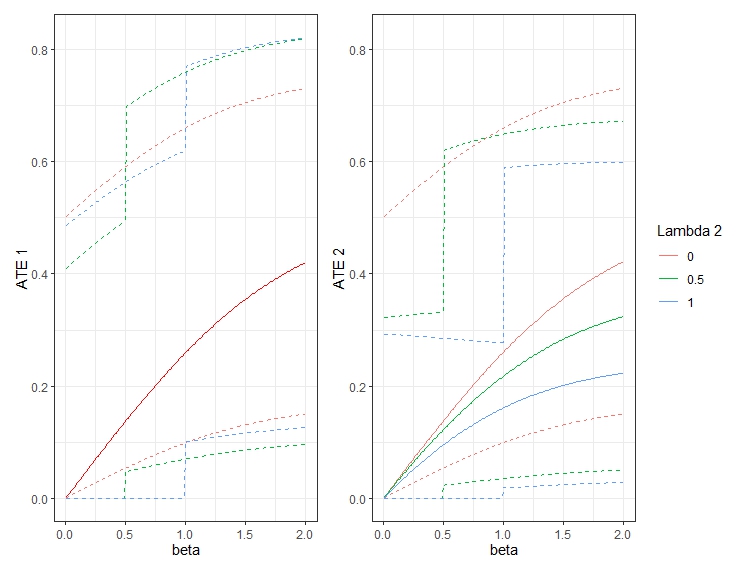}

\caption{Bounds for the ATE for a probit model with discrete regressors and
time effects at different time periods. The solid lines show the true
ATE at time 1 (left panel) and time 2 (right panel), while the dotted
lines show the bounds on the ATEs. The colors correspond to the three
different values of $\lambda_{2}.$ The dotted lines corresponding
to $\lambda_{2}=0$ in the left panel correspond to the bounds for
the binary probit model without time effects and with discrete regressors
as in \citet{chernozhukov2013average}. The bounds corresponding to
$\lambda_{2}\protect\neq0$ correspond to the bounds for the binary
probit with time effects and discrete regressors computed via Theorem
\ref{thm:intersection_bound_strict_exogeneity}.}

\label{probit_CFHN_time}
\end{figure}

\end{document}